\documentclass[12pt,twoside,leqno]{article}

\usepackage{amsmath,amssymb,amsthm}
\usepackage[margin=1in]{geometry}
\usepackage{enumitem}
\usepackage{nicefrac}
\usepackage{url}
\usepackage{mathrsfs}
\usepackage{fancyhdr}
\usepackage{graphics}
\usepackage{filecontents}

\newtheorem{theorem}{Theorem}
\newtheorem{lemma}[theorem]{Lemma}
\newtheorem{corollary}[theorem]{Corollary}

\theoremstyle{remark}
\newtheorem{observation}[theorem]{Observation}
\newtheorem{remark}[theorem]{Remark}
\newtheorem{example}[theorem]{Example}
\newtheorem{conjecture}[theorem]{Conjecture}
\newtheorem*{claim}{Claim}

\newcommand{\half}{\nicefrac{1}{2}}

\newcommand{\codomainall}{\mathcal{U}}
\def\codomainbool{\mathcal{B}}
\def\codomainnneg{\codomainbool}
\def\codomainpos{\codomainbool^{>0}}
\def\codomainreal{\mathscr{B}}
\def\codomaineff{\codomainbool^{p}}

\def\codomainup{\codomainbool^{\mathrm{up}}}
\def\codomaindown{\codomainbool^{\mathrm{down}}}
\def\codomainupeff{\codomainbool^{\mathrm{up},p}}
\def\codomaindowneff{\codomainbool^{\mathrm{down},p}}

\makeatletter
\def\prob#1#2#3{\goodbreak\begin{list}{}{\labelwidth\z@ \itemindent-\leftmargin
                        \itemsep\z@  \topsep6\p@\@plus6\p@
                        \let\makelabel\descriptionlabel}
                \item[\it Name]#1
                \item[\it Instance]#2
                \item[\it Output]#3
                \end{list}}
\makeatother

\def\mumin{\mu_{\min}}
\def\mumax{\mu_{\max}}
\def\numin{\nu_{\min}}
\def\numax{\nu_{\max}}

\def\constzero{\mathbf{0}}
\def\constone{\mathbf{1}}
\def\unaryzero{\unary_0}
\def\unaryone{\unary_1}
\def\unary{\delta}

\def\RHPi{\mathrm{\#RH}\Pi_1}

\def\fclone#1{\langle#1\rangle}
\def\rclone#1{\langle#1\rangle_{\mathrm{R}}}
\def\fclonelim#1{\langle#1\rangle_\omega}
\def\fclonelimeff#1{\langle#1\rangle_{\omega,p}}
\def\NP{\mathsf{NP}}
\def\RP{\mathsf{RP}}

\def\numP{\mathsf{\#P}}
\def\nCSP{\mathrm{\#\mathsf{CSP}}}
\def\CSP{\mathsf{CSP}}
\def\LSM{\mathsf{LSM}}
\def\LSMp{\mathsf{LSM}^{>0}}
\def\BIS{\mathsf{\#BIS}}
\def\SAT{\mathsf{\#SAT}}

\def\B{\{0,1\}}
\def\R{\mathbb{R}}
\def\Q{\mathbb{Q}}

\def\Rnonneg{\R^{\geq0}}
\def\Cnum{\mathbb{C}}
\def\Rpos{\R^{>0}}
\def\Reff{\R^p}
\def\Nat{\mathbb{N}}
\def\poly{\mathop\mathrm{poly}\nolimits}
\def\calA{\mathcal A}
\def\calC{\mathcal C}
\def\calF{\mathcal F}
\def\calR{\mathcal R}

\def\calM{\mathcal M}
\def\calP{\mathcal P}
\def\Fhat{\widehat F}
\def\Fhat{\widehat F}
\def\Ghat{\widehat G}
\def\Hhat{\widehat H}
\def\hatF{\Fhat}

\def\vecc{\boldsymbol c}

\def\vecw{\boldsymbol w}
\def\vecx{\boldsymbol x}
\def\vecy{\boldsymbol y}
\def\vecz{\boldsymbol z}
\renewcommand{\Bar}[1]{\overline{#1}}
\newcommand{\vecxny}{\vecx\wedge\vecy}
\newcommand{\vecxuy}{\vecx\vee\vecy}
\newcommand{\ind}[1][A]{\constone_A}
\def\OR{\mathrm{OR}}
\def\EQ{\mathrm{EQ}}
\def\IMP{\mathrm{IMP}}
\def\NEQ{\mathrm{NEQ}}

\def\NAND{\mathrm{NAND}}
\def\maxF{F_{\max}}
\def\Fmax{\maxF}
\def\APred{\leq_\mathrm{AP}}
\def\APeq{=_\mathrm{AP}}
\def\dom{D}
\def\ring{\mathscr{C}}
\def\IDo{\mathrm{ID}_1}
\def\IRt{\mathrm{IR}_2}
\def\ILt{\mathrm{IL}_2}

\def\omegadef{$\mathrm{pps}_\omega$-definable}
\def\ppdef{pps-definable}
\def\ppsdefinability{$\mathrm{pps}$-definability}
\def\ppdefinable{\ppdef}
\def\ppformula{pps-formula}
\def\ppformulas{pps-formulas}
\def\omegadefinition{$\mathrm{pps}_\omega$-definition}
\def\omegadefinable{\omegadef}
\def\omegadefinability{$\mathrm{pps}_\omega$-definability}

\let\epsilon=\varepsilon
\let\phi=\varphi
\let\rho=\varrho

\def\EQp{\mathrm{EQ}'}
 \def\ZIsing{Z_\mathrm{Ising}}
 \def\calC{\mathcal{C}}

\begin{document}

\begin{titlepage}
\renewcommand{\thefootnote}{\fnsymbol{footnote}}
\begin{center}
{\Large\textsc{\strut
The expressibility of functions on the Boolean domain, with applications to
Counting CSPs\footnote{The work reported in this paper was supported by an EPSRC Research Grant 
``Computational Counting'' (refs. EP/I011528/1,  EP/I011935/1, EP/I012087/1), and by an NSERC Discovery Grant, 
and by an EPSRC doctoral training grant. Part of the work was supported by a visit to the Isaac Newton Institute 
for Mathematical Sciences, under the programme ``Discrete Analysis''.
Some of the results were announced in the preliminary
papers~\cite{BuDyGJ12} and \cite{McQuil11}.}}}\\\vspace{0.8cm}
\renewcommand{\thefootnote}{\arabic{footnote}}\setcounter{footnote}{0}
{\large
Andrei A. Bulatov\footnote{School of Computing Science,
Simon Fraser University, 8888 University Drive, Burnaby BC, V5A 1S6, Canada.}\hspace{1cm}
Martin Dyer\footnote{School of Computing,
University of Leeds, Leeds LS2~9JT, United Kingdom.}\\[0.25\baselineskip]
Leslie Ann Goldberg\footnote{Department of Computer Science,
Ashton Building, University of Liverpool, Liverpool L69 3BX,
United Kingdom.}\hspace{6mm}
Mark Jerrum\footnote{School of Mathematical Sciences,
Queen Mary, University of London, Mile End Road, London E1 NS,
United Kingdom.}\hspace{6mm}
Colin McQuillan\footnotemark[3]\\\vspace{1cm}
\today}
\end{center}\vspace{0.8cm}

\begin{abstract}
An important tool  in the study of the complexity of Constraint Satisfaction Problems (CSPs) is
the notion of a relational clone, which is the set of all relations expressible
using primitive positive formulas over a particular set of base relations.
Post's lattice gives a complete classification of all
Boolean
relational clones, and this has been used to
classify the computational difficulty of CSPs.
Motivated by a desire to understand the computational complexity of (weighted) counting CSPs,
we develop an analogous notion of functional clones and study the landscape of these clones.
One of these clones is the collection of log-supermodular (lsm) functions, which turns
out to play a significant role in classifying counting CSPs.
In the conservative case (where  all nonnegative unary functions are available),
we show that there are no functional clones lying strictly between
the clone of lsm functions and the total clone (containing all functions).  Thus, any counting
CSP that contains a single nontrivial non-lsm function is computationally as hard
to approximate as any problem in \#P.
Furthermore, we show that any non-trivial functional clone (in a sense that will be made precise) contains the
binary function ``implies''.  As a consequence, in the conservative case, all non-trivial counting CSPs are  as hard
as $\BIS$, the  problem of counting independent sets in a  bipartite graph.
Given the complexity-theoretic results, it is natural to ask whether the ``implies'' clone is equivalent to
the clone of lsm functions. We use the M\"obius transform and the Fourier transform to show that these clones coincide precisely up to arity 3. It is an intriguing open question whether the lsm clone is finitely generated.
Finally, we investigate functional clones in which only restricted classes of
unary  functions are available.
\end{abstract}
\end{titlepage}

\setcounter{footnote}{0}
\section{Introduction}\label{sec:intro}

In the classical setting, a
(non-uniform)
constraint satisfaction problem $\CSP(\varGamma)$ is specified by a finite domain~$D$ and a constraint language $\varGamma$, which is a set of relations of varying arities over~$D$. For example, $D$ might be the Boolean domain $\{0,1\}$ and $\varGamma$ might be the set containing the single relation $\NAND = \{(0,0),(0,1),(1,0)\}$. An instance of $\CSP(\varGamma)$ is a set of $n$ variables taking values in~$D$, together with a set of constraints on those variables.  Each constraint is a relation $R$ from~$\varGamma$ applied to a tuple of variables, which is called the ``scope'' of the constraint.  The problem is to find an assignment of domain elements to the variables which satisfies all of the constraints. For example, the problem of finding an independent set in a graph can be represented as a CSP with $\varGamma=\{\NAND\}$. The vertices of the graphs are the variables of the CSP instance. The instance has one $\NAND$ constraint for each edge of the graph. Vertices whose variables are mapped to domain element~$1$ are deemed to be in the independent set. Constraint satisfaction problems (CSPs) may be viewed as generalised satisfiability problems, among which usual satisfiability is a very special case.

The notion of expressibility is key to understanding the complexity of CSPs.
A \emph{primitive positive formula} (pp-formula)
in variables $V =\{v_1,\ldots,v_n\}$
is a formula of the form
$$\exists v_{n+1} \ldots v_{n+m}\, \bigwedge_i \phi_i,$$
where each atomic formula $\phi_i$ is
either a relation $R$ from $\varGamma$ applied to some of the variables in~$V'=\{v_1,\ldots,v_{n+m}\}$
or an equality relation of the form $v_i=v_j$,
which we write as $\EQ(v_i,v_j)$.
For example, the formula
$$\exists v_3\,\NAND(v_1,v_2)\EQ(v_1,v_3)\EQ(v_2,v_3)$$
is a pp-formula in variables $v_1$ and $v_2$.
This formula corresponds to the relation $\{(0,0)\}$ since the only
way to satisfy the constraints is to map both $v_1$ and $v_2$ to the domain element~$0$.

The \emph{relational clone} $\rclone \varGamma$ is the set of all relations
expressible as pp-formulas over~$\varGamma$.
Relational clones have played a key role in the development of the complexity of CSPs
because of the following important fact, which is described,
for example, in the expository chapter of Cohen and
Jeavons~\cite{CohenJeavons}: If two sets of relations
$\varGamma$ and~$\varGamma'$ generate the same relational clone,
then the computational complexities
of the corresponding CSPs, $\CSP(\varGamma)$ and $\CSP(\varGamma')$, are
exactly the same. Thus,
in order to understand the complexity of CSPs, one
does not to consider all sets of
relations~$\varGamma$. It suffices to  consider those that are
relational clones.

Recently, there has been considerable interest in the computational
complexity of counting CSPs (see, for example~\cite{BDGJJR12,   CCL11, Chen11, wbool, trichotomy,   Tomo}).  Here, the goal is to count the number of
solutions of a CSP rather than merely to decide if a solution exists.  In fact, in order to
encompass the computation of partition functions of models from statistical
physics and other generating functions, it is  common (see, for example, \cite{BG}) to consider weighted
sums, which can be expressed by replacing the relations in the constraint
language by real-valued or complex-valued functions.
In this case, the weight of an
assignment (of domain values to the variables)
is the product of the function values corresponding to that
assignment, while the value of the CSP instance itself is the sum of the
weights of all assignments.
If $I$ is an instance of such a counting CSP
then this weighted sum is called the ``partition function of~$I$'' (by analogy
with the concept in statistical physics) and is denoted~$Z(I)$. For a
finite set of functions~$\varGamma$ we are interested in the problem
$\nCSP(\varGamma)$, which is the problem of computing~$Z(I)$, given an instance~$I$  which uses only functions
from~$\varGamma$.

Our first goal (see \S\ref{sec:functionalClones}) is
to determine the most useful analogues of pp-definability and relational clones   in the context of (weighted) counting
CSPs (\#CSPs), and to see what insight  this provides into the
computational complexity of these problems.
It is clear that, in order to adapt  the  concept of pp-definability to
the counting setting, one should  replace a conjunction of
relations by a product of functions  and
replace existential quantification by summation.
However, there are sensible alternatives for the detailed definitions,
and these have ramifications for the complexity-theoretic consequences.
There is at least one proposal in the literature for extending
pp-definability to the algebraic/functional setting --- that of Yamakami~\cite{Tomo}.
However, we find it useful to adopt a more liberal notion of pp-definability,
including a limit operation.
Without this, a functional clone could contain arbitrarily close approximations to a function~$F$ of interest, without including~$F$ itself.
We call this analogue of pp-definability   ``\omegadefinability''.
The notion of \omegadefinability\ leads to a more  inclusive functional clone than the one considered in~\cite{Tomo}.

Aside from a desire for tidiness, there is a good empirical motivation for
introducing limits. Just as pp-definability is closely related to
polynomial-time reductions between classical CSPs, so is \omegadefinability{} related to approximation-preserving reductions between (weighted) counting CSPs.
Lemma~\ref{obsAPclones} is a precise statement of this connection.  Many
approx\-imation-preserving reductions in the literature (for example, those in~\cite{GJ07}) are based not on a fixed ``gadget'' but on sequences of
increasingly-large gadgets that come arbitrarily close to some property
without actually attaining it. Our notion of \omegadefinability{} is intended to capture this phenomenon.

The second, more concrete goal of this paper (see \S\ref{sec:relclones}--\S\ref{sec:main})
is to explore the space of functional clones
and to use what we learn about this space to classify
the complexity of approximating \#CSPs.  We
restrict attention to the Boolean situation so the domain is $\B$ and
the allowed functions are of the form $\B^k\to\Rnonneg$ for some integer~$k$.
We examine the landscape of
functional clones
for the case in which all nonnegative unary functions
(weights) are available.
This case is known as the \emph{conservative} case. It is also
studied in the context of decision and optimisation CSPs~\cite{Bulato11,KolZiv12}
and in work related to counting CSPs
such as Cai, Lu and Xia's work on classifying ``$\mathrm{Holant}^*$''
problems~\cite{CLX11}.
The conservative case is easier to to classify than the general
case, so we are able to  construct a useful map of the landscape of functional
clones  (see Theorem~\ref{thm:main}).
Note that Yamakami~\cite{Tomo} has considered an even more special case in which
all unary weights (including negative weights) are available.
In that case the landscape turns out to be less rich and more pessimistic ---
negative weights introduce cancellation, which tends to drive approximate
counting CSPs in the direction of intractability.

An issue that turns out to be important in the classification of conservative functional
clones is \emph{log-supermodularity}.
Roughly, a function with Boolean domain is said to be log-supermodular if its logarithm is
supermodular. (A formal definition appears later.)
It is a non-trivial fact
(Lemma~\ref{lem:lsmClosure}) that the set $\LSM$ of log-supermodular
 functions is  a functional clone (using our notion of \omegadefinability).

Conservative functional clones are classified as follows.
 A particularly simple functional clone is the
 clone  generated by the disequality relation.
A counting CSP derived from this
clone is trivial to solve exactly, as the partition function factorises.   We say that functions from this clone are of ``product form''.  Our main
result (Theorem~\ref{thm:main}) is that any clone that contains a
function~$F$ that is not of product form necessarily contains
the   binary relation
$\IMP=\{(0,0),(0,1),(1,1)\}$.
This has important complexity-theoretic consequences, which will be discussed presently.
Furthermore, (also Theorem~\ref{thm:main}),
any non-trivial
clone that contains a function~$F$ that is  not log-supermodular
actually
contains all functions.
Therefore a large part of the functional
clone landscape is very simple. In particular, we have a complete understanding of
the clones below the clone generated by $\IMP$ and
of the clones above $\LSM$.
The complexity of the landscape of functional clones is thus sandwiched between the clone generated by~$\IMP$ and the clone $\LSM$.

In order to derive complexity-theoretic consequences (see Theorem~\ref{corthm:main}),
we also present an efficient version of \omegadefinability, and a corresponding notion of functional clone.
The complexity-theoretic consequences are  the third contribution of the paper.
In order to describe these, we need a quick digression into the complexity of approximate counting.
The complexity class $\RHPi$
of counting problems was introduced by
Dyer,  Goldberg, Greenhill and Jerrum~\cite{trichotomy} as a means
to classify a wide class of approximate counting problems that were
previously of indeterminate computational complexity.  The problems in $\RHPi$
are those that
can be expressed in terms of counting the number of models of a logical formula
from a certain syntactically restricted class.
The complexity class $\RHPi$ has a completeness class (with respect to
approximation-preserving ``AP-reductions'')
which includes a wide and ever-increasing
range of natural counting problems, including:
independent sets in a bipartite graph,
downsets in a partial order,
configurations in the Widom-Rowlinson model (all~\cite{trichotomy})
and stable matchings~\cite{stablematchings}.
Either all of these problems admit a Fully Polynomial Randomised
Approximation Scheme (FPRAS), or none do.
The latter is conjectured. A typical complete problem in this class
is $\BIS$, the problem of  counting independent sets
in a bipartite graph.

Our complexity-theoretic results
are presented in~\S\ref{sec:complexity}.
As noted above, $\nCSP(\calF)$ is computationally easy if every function in~$\calF$ is of product form.
We show that, in every other (conservative) case, it is as difficult to approximate as $\BIS$.
If, in addition, $\calF$ contains a function~$F$ which is not   log-supermodular, then
the counting problem $\nCSP(\calF)$ turns out to be universal for Boolean counting CSPs and hence is provably $\NP$-hard to approximate. As immediate corollaries, we recover existing results concerning the complexity of computing the partition function of the Ising model~\cite{GJ07}.

Given the above discussion, one might speculate that the IMP-clone and $\LSM$ are the same.  In fact, they are not. In \S\ref{sec:cloneLSMk}, we examine the classes $\fclone{\LSM_k}$ generated by lsm functions of arity at most $k$. We show that $\fclone{\LSM_3}$ is equal to the $\IMP$-clone, but we give a proof  that
$\fclone{\LSM_3}$ is strictly contained in $\fclone{\LSM_4}$. This mirrors the situation for VCSPs, where binary submodular functions can express all ternary submodular functions but not all arity 4 submodular functions \cite{ZCJ09}. However, we do not know whether
there is a fixed~$k$ such that
$\LSM=\fclone{\LSM_k}$. We conjecture that this is not the case. If $\LSM=\fclone{\LSM_k}$ for some $k$, there would still remain the question of whether $\LSM$ is finitely generated, i.e., whether it is the functional clone generated by some finite set of functions $\calF$. We conjecture the opposite, that there is no such $\calF$.

Finally, in \S\ref{sec:restrictedWeights} and \S\ref{sec:fewWeights}, we step outside the conservative case, and study functional clones in which only restricted classes of unary functions are available.  As might be expected, this yields a richer structure of functional clones, including one that corresponds to the ferromagnetic Ising model with a consistent field, a problem that is tractable in the FPRAS sense~\cite{JS93}.  We also exhibit in this setting two clones that are provably incomparable with respect to inclusion, even though their corresponding counting CSPs are related by approximation-preserving reducibility.  These counting CSPs have natural interpretations as (a)~evaluating the weight enumerator of a binary code, and (b)~counting independent sets in a bipartite graph.  This example shows that, in demonstrating intractability, it may sometimes be necessary to use reductions that go beyond \omegadefinability.

Although we focus on approximation of the partition functions of (weighted)
\#CSPs in this paper, there is, of course, an extensive literature on exact
evaluation.  See, for example, the recent survey of Chen~\cite{Chen11}.

\section{Functional clones}\label{sec:functionalClones}
As usual, we denote the natural numbers by $\Nat$, the real numbers by $\R$ and the complex numbers by $\Cnum$. For $n\in\Nat$, we denote the set $\{1,2,\ldots,n\}$ by $[n]$.

Let $(\ring,+,\times)$ be any subsemiring of $(\Cnum,+,\times)$, and
let~$\dom$ be a finite domain. For $n\in\Nat$, denote by~$\codomainall_n$ the set of all functions $\dom^n\to\ring$; also denote by $\codomainall =\codomainall_0\cup\codomainall_1\cup\codomainall_2\cup\cdots$
the set of functions of all arities.
Note that we do not specify the domain, which we take to be understood from the context, in this notation.
Suppose $\calF\subseteq\codomainall$ is some collection of functions,
$V=\{v_1,\ldots,v_n\}$ is a set of variables and  $\vecx:\{v_1,\ldots,v_n\}\to\dom$ is
an assignment to those variables.
An atomic formula has the form
$\phi=G(v_{i_1},\ldots,v_{i_a})$ where $G\in\calF$, $a=a(G)$ is the arity of~$G$,
and $(v_{i_1},v_{i_2},\ldots,v_{i_a})\in V^a$
is a scope.  Note that repeated variables are allowed.
The function $F_\phi:\dom^n\to\ring$
represented by the atomic formula $\phi=G(v_{i_1},\ldots,v_{i_a})$ is just
$$
F_\phi(\vecx)=G(\vecx(v_{i_1}),\ldots,\vecx(v_{i_a}))=G(x_{i_1},\ldots,x_{i_a}),
$$
where from now on we write $x_j=\vecx(v_j)$.

A \ppformula{} (``primitive product summation formula'') is a summation\footnote{To avoid ambiguity, we will try to use ``summation'' of functions only with the meaning given here. Sums of different functions will be referred to as ``addition''.} of a product of atomic formulas.
A \ppformula~$\psi$ in variables $V\subseteq V'=\{v_1,\ldots,v_{n+m}\}$ over~$\calF$
has the form
\begin{equation}\label{eq:ppformulaDef}
\psi=\sum_{v_{n+1},\ldots,v_{n+m}}\,\prod_{j=1}^s\phi_j,
\end{equation}
where $\phi_j$ are all atomic formulas over~$\calF$ in the variables~$V'$.
(The variables $V$ are free, and the others, $V'\setminus V$, are bound.)
The formula~$\psi$ specifies a function $F_\psi:\dom^n\to\ring$
in the following way:
\begin{equation}\label{eq:ppfunctionDef}
F_\psi(\vecx)=\sum_{\vecy\in\dom^m}\prod_{j=1}^sF_{\phi_j}(\vecx,\vecy),
\end{equation}
where $\vecx$ and $\vecy$ are assignments
$V\to\dom$,
$V'\setminus V\to\dom$.
The functional clone~$\fclone\calF$ generated by~$\calF$
is the set of all  functions in $\codomainall$
that can be represented
by a \ppformula{} over~$\calF\cup\{\EQ\}$
where $\EQ$ is the binary equality function
defined by $\EQ(x,x)=1$ and $\EQ(x,y)=0$ for $x\not=y$.
We refer to the \ppformula{} as an \emph{implementation} of the function.

Since \ppformula{s} are defined using sums of products (with just one level of each),
we need to check that functions that are \ppdef{} in terms of functions
that are themselves \ppdef{} over~$\calF$ are actually directly \ppdef{} over~$\calF$.
The following lemma ensures that this is the case.

\begin{lemma}
\label{lem:fcloneTransitive}
If $G\in\fclone\calF$ then $\fclone{\calF,G}=\fclone{\calF}$.
\end{lemma}

Note that, to simplify notation, we write $\fclone{\calF,G}$ in place of the more correct $\fclone{\calF\cup\{G\}}$.  More generally, we shall often
drop set-brackets, replace the union symbol~$\cup$ by a comma, and confuse
a singleton set with the element it contains.

\begin{proof}[Proof of Lemma~\ref{lem:fcloneTransitive}]
Let $\calF' = \calF \cup \{\EQ\}$.
Suppose  that $\psi$ is  a \ppformula{} over $\calF'\cup\{G\}$ given by
\begin{equation}\label{eqone}
\psi = \sum_{v_{n+1},\ldots,v_{n+m}} \left(\prod_{i=1}^{r} \phi_i\right)\left( \prod_{j=1}^s\psi_j\right),\end{equation}
where  $\{\phi_i\}$ are atomic $\calF'$-formulas and $\{\psi_j\}$ are atomic $G$-formulas
in the variables $V'$.
Then
\begin{equation}
\label{eqtwo}
F_\psi(\vecx)=\sum_{\vecy\in\dom^m}\left(\prod_{i=1}^{r}F_{\phi_i}(\vecx,\vecy)\right)\left(\prod_{j=1}^sF_{\psi_j}(\vecx,\vecy)\right),
\end{equation}
where $\vecx$ and $\vecy$ are assignments
$
\vecx:\{v_1,\ldots,v_n\}\to\dom$ and
 $\vecy:\{v_{n+1},\ldots,\allowbreak v_{n+m}\}\to\dom$.
Now, since $G$ is \ppdef{} over~$\calF'$,
and each~$\psi_j$ is an atomic $G$-formula in the variables $V'$, we can
write each
$\psi_j$ as
$$\psi_j = \sum_{v_{\nu_j+1},\ldots,v_{\nu_j+\ell}}\,\prod_{k=1}^t\phi_{j,k},$$
where $\ell$ is the number of bound variables used in the definition of $\psi_j$
($\ell$ is independent of~$j$),
$\nu_j= n+m+(j-1)\ell$
is the number of free variables plus the number of bound variables that are ``used up'' by
$\psi_1,\ldots,\psi_{j-1}$, and
each $\phi_{j,k}$ is an atomic $\calF' $-formula
over the variables $V' \cup \{v_{\nu_j+1},\ldots,v_{\nu_j+\ell}\}$.
We get
\begin{align*}
F_{\psi}(\vecx) &=\sum_{\vecy\in\dom^{m}}\left(\prod_{i=1}^{r}F_{\phi_i}(\vecx,\vecy)\right)
\left(\prod_{j=1}^s
\sum_{\vecz^j\in\dom^\ell}\prod_{k=1}^tF_{\phi_{j,k}}(\vecx,\vecy,\vecz^j)\right)\\
&=\sum_{\vecy\in\dom^{m}}\sum_{\vecz^1\in\dom^\ell}\cdots\sum_{\vecz^{s}\in\dom^\ell}
\left(\prod_{i=1}^{r}F_{\phi_i}(\vecx,\vecy)\right)\left(\prod_{j=1}^s\prod_{k=1}^tF_{\phi_{j,k}}(\vecx,\vecy,\vecz^j)\right),
\end{align*}
where each $\vecz^j$ is an assignment $\vecz^j:\{v_{\nu_j+1},\ldots,v_{\nu_j+\ell}\}\to\dom$.  So
$$ \phi=
\sum_{v_{n+1},\ldots,v_{\nu_{s}+\ell}}\,\left(\prod_{i=1}^{r} \phi_i\right)\left( \prod_{j=1}^s
\prod_{k=1}^t\phi_{j,k}\right)$$
is a   \ppformula{} over~$\calF'$ for the function~$F_\psi$.
\end{proof}

To extend the notion of definability, we allow limits as follows.
We say that an $a$-ary function $F\in\codomainall$ is \omegadef{} over $\calF$ if there exists a finite subset~$S_F$
of $\calF$, such that, for every $\epsilon>0$, there exists an $a$-ary function $\Fhat$, pps-definable over $S_F$, such that
$$\|\Fhat-F\|_\infty=\max_{\vecx\in\dom^a}|\Fhat(\vecx)-F(\vecx)|<\epsilon.$$

Denote the set of functions that are \omegadef{} over~$\calF\cup\{\EQ\}$
by~$\fclonelim\calF$;  we call this the {\it \omegadef{} functional clone\/} generated by~$\calF$. Observe that functions in $\fclonelim\calF$ are determined only by finite subsets of~$\calF$.
Also, some  functions taking values outside~$\ring$
may be the limit of functions pps-definable over $\calF$. But they are not \omegadef{}, since the function values of the limit must be in~$\ring$.  The domain $\ring$ of the universal class of functions~$\codomainall$ in operation at any time will be clear from the context.

The following lemma is an analogue of Lemma~\ref{lem:fcloneTransitive}.
\begin{lemma}
\label{lem:fclonelimTransitive}
If $G\in\fclonelim\calF$ then $\fclonelim{\calF,G}=\fclonelim{\calF}$.
\end{lemma}

\begin{proof}
Let $\calF' = \calF \cup \{\EQ\}$. Suppose that $H$ is an $a$-ary function in $\fclonelim{\calF,G}$.
Let $S_H$ be a finite subset of $\calF'\cup\{G\}$ such that the following is
true:
Given $\epsilon>0$, there exists an $a$-ary function
$\Hhat$, pps-definable over~$S_H$,
such that
$\|\Hhat-H\|_\infty<\epsilon/2$.  Let $\psi$ be a \ppformula{}
over $S_H$ representing~$\Hhat$.
For any function~$\Ghat$ with the same arity as~$G$, denote by $\psi[G{:=}\Ghat]$ the formula
obtained from~$\psi$ by replacing all occurrences of $G$ by~$\Ghat$.  By continuity of the operators of \ppformula{s},
we know there exists $\delta>0$ such that, for every function~$\Ghat$ of
the same arity as~$G$,
$\|\Ghat-G\|_\infty<\delta$ implies
$$\|F_{\psi[G{:=}\Ghat]}-F_\psi\|_\infty<\epsilon/2.$$
This claim will be explicitly quantified in the proof of Lemma~\ref{lem:fclonelimeffTransitive}, but we don't
need so much detail here.
Of course, $\Hhat = F_\psi$ so for each such $\Ghat$ we have
$\|F_{\psi[G{:=}\Ghat]}-\Hhat\|_\infty<\epsilon/2$.
Now
let $S_G$ be the finite subset of $\calF'$ used to show that $G$ is \omegadef{} over $\calF'$.
Let $S = S_G \cup S_H \setminus \{G\} \subseteq \calF'$.
 Choose a function
 $\Ghat$, pps-definable over~$S_G$,
 satisfying $\|\Ghat-G\|_\infty<\delta$.  Notice that
 $\Ghat \in \fclone{S}$
 and
 $F_{\psi}\in \fclone{S,G}$ so
$F_{\psi[G{:=}\Ghat]}\in\fclone  S$ (by Lemma~\ref{lem:fcloneTransitive}), and
$$\|F_{\psi[G{:=}\Ghat]}-H\|_\infty \leq \|F_{\psi[G{:=}\Ghat]}-\Hhat\|_\infty+\|\Hhat-H\|_\infty<\epsilon.$$
Since $\epsilon>0$ is arbitrary,
and $S\subseteq\calF'$ is finite, we conclude that $H\in\fclonelim{\calF}$.
\end{proof}

That completes the setup for expressibility.  In order to deduce complexity results,
we need an efficient version of $\fclonelim\calF$.
We say that a function~$F$ is \emph{efficiently} \omegadef{} over $\calF$ if
there is a finite subset $S_F$ of
$\calF$,
and a TM $\calM_{F,S_F}$ with the following property:
on input $\epsilon>0$, $\calM_{F,S_F}$ computes
a \ppformula~$\psi$ over $S_F$
such that $F_\psi$ has the same arity as~$F$ and
  $\| F_\psi-F\|_\infty<\epsilon$.
The running time of $\calM_{F,S_F}$ is at most a polynomial in  $\log\epsilon^{-1}$.
Denote the set of functions in $\codomainall$ that are efficiently \omegadef{} over~$\calF\cup\{\EQ\}$
by~$\fclonelimeff\calF$;  we call this the {\it efficient \omegadef{} functional clone\/}
generated by~$\calF$,

The following useful observation is immediate from the definition
of $\fclonelimeff{\calF}$.
\begin{observation}
\label{obs:26July}
Suppose $F\in \fclonelimeff{\calF}$ (or $F\in \fclonelim{\calF}$). Then there is a finite subset~$S_F$ of~$\calF$ such
that $F\in\fclonelimeff{S_F}$ (resp.\ $F\in\fclonelim{S_F}$).
\end{observation}

The following lemma is an analogue of Lemma~\ref{lem:fclonelimTransitive} .
\begin{lemma}
\label{lem:fclonelimeffTransitive}
If $G\in\fclonelimeff\calF$ then $\fclonelimeff{\calF,G}=\fclonelimeff{\calF}$.
\end{lemma}

\begin{proof} Let $\calF' = \calF \cup \{\EQ\}$.
Suppose $H$ is an $a$-ary function in $\fclonelimeff{\calF',G}$.
Our goal is to
specify a finite subset~$S$ of $\calF'$ and to
construct a TM
$M_{H,S}$
with the following property:
on input $\epsilon>0$, $M_{H,S}$  should compute
an $a$-ary \ppformula~$\phi$ over $S$
such that
  $\| F_{\phi}-H\|_\infty<\epsilon$.
The running time of $M_{H,S}$  should be at most a polynomial in  $\log\epsilon^{-1}$.

Let $S_H$ be the finite subset of $\calF'\cup\{G\}$ from the efficient \omegadefinition{} of~$H$ over $\calF'\cup\{G\}$.
 Given
an input $\epsilon/2$, the TM $M_{H,S_H}$
computes an $a$-ary \ppformula~$\psi$ over $S_H$
such that
$\| F_\psi-H\|_\infty<\epsilon/2$. Write $\psi$ as
in Equation~(\ref{eqone}) so $F_\psi$ is written as in Equation~(\ref{eqtwo}).
Suppose that,
for $j\in[s]$ and $\vecy\in\B^m$,
$\delta_{j,\vecy}(\vecx)$ is a function of $\vecx$.
Consider the expression
\begin{align*}\Upsilon(\vecx) &= \sum_{\vecy\in\dom^m}\left(\prod_{i=1}^{r}F_{\phi_i}(\vecx,\vecy)\right)\left(\prod_{j=1}^s (F_{\psi_j}(\vecx,\vecy) + \delta_{j,\vecy}(\vecx))\right)\\
&\qquad- \sum_{\vecy\in\dom^m}\left(\prod_{i=1}^{r}F_{\phi_i}(\vecx,\vecy)\right)\left(\prod_{j=1}^sF_{\psi_j}(\vecx,\vecy)\right),
\end{align*}
which can be expanded as
$$\Upsilon(\vecx) = \sum_{\vecy\in\dom^m}
\sum_{\emptyset \subset T\subseteq [s]} C_{\vecy,T}(\vecx) \prod_{j\in T} \delta_{j,\vecy}(\vecx),$$
where
$$C_{\vecy,T}(\vecx) =   \prod_{i=1}^{r}F_{\phi_i}(\vecx,\vecy)
\prod_{j \in [s]\setminus T} F_{\psi_j}(\vecx,\vecy).$$
Let $C = \max_{\vecx,\vecy,T} |C_{\vecy,T}(\vecx)|$ and let $\delta= \epsilon 2^{-(s+1)} |D|^{-m} C^{-1}<1$.

Now
let $S_G$ be the finite subset of $\calF'$ used to show that $G$ is
efficiently \omegadefinable{} over $\calF'$.
Given the input $\delta$, the TM $M_{G,S_G}$
computes a \ppformula{} $\widehat\psi$ over $S_G$ representing a function $F_{\widehat\psi}$
with the same arity as~$G$
such that $\|F_{\widehat\psi} -  G\|_\infty < \delta$.
Since each $\psi_j$ is an atomic $G$-formula in the variables~$V'$,
we may appropriately name the variables of $\widehat\psi$ to obtain
a \ppformula{} $\widehat\psi_j$ over $S_G$
such that  $\|F_{\widehat\psi_j} - F_{\psi_j}\|_\infty < \delta$.

For $\vecy\in\dom^m$, let $\delta_{j,\vecy} (\vecx)= F_{\widehat\psi_j}(\vecx,\vecy) - F_{\psi_j}(\vecx,\vecy)$
and note that $|\delta_{j,\vecy}(\vecx)| \leq \delta$.
Let $S = S_G \cup S_H \setminus\{G\} \subseteq \calF'$.
 Let $\psi'$ be the formula over $S$ formed from~$\psi$ by substituting each occurrence
of $\psi_j$ with $\widehat\psi_j$.
Let $\phi$ be the \ppformula{} over $S$ for
the function $F_{\psi'}$
which is constructed as in the proof of Lemma~\ref{lem:fcloneTransitive}.
From the calculation above,
\begin{align*}\|F_\phi - F_{\psi}\|_\infty &=
\|F_{\psi'} - F_{\psi}\|_\infty \\
&= \max_{\vecx} |F_{\psi'}(\vecx) - F_{\psi}(\vecx)|\\
&=\max_{\vecx} |\Upsilon(\vecx)|\\
&\leq 2^s |D|^m C \delta = \epsilon/2.
\end{align*}

Note that $$\|F_\phi - H\|_\infty \leq
\|F_\phi - F_\psi\|_\infty
+
\|F_\psi-H\|_\infty
< \epsilon.$$
Thus, the formula $\phi$ is an appropriate output for our TM $M_{H,S}$.

Finally, let us check how long the computation takes.
The running time of $M_{H,S_H}$ is at most $\poly(\log\epsilon^{-1})$.
Since this machine outputs the formula~$\psi$, we conclude that $m$ and $r$ and $s$
are bounded from above by polynomials in $\log\epsilon^{-1}$.
Let $\varDelta$ be the
ceiling of the maximum absolute value of any function in $S_H$.
Note that $C \leq \varDelta^{r+s}$.
The running time of $M_{G,S_G}$ is at most $\poly(\log(\delta^{-1}))$, which
is at most a polynomial in $m + s + \log(C) + \log(\epsilon^{-1})$ which
is at most a polynomial in $\log(\epsilon^{-1})$.
Finally, the direct manipulation of the formulas that we did
(renaming variables from $\widehat\psi$ to obtain
$\widehat\psi_j$ and producing the \ppformula{} $\phi$ from $\psi$ and the $\widehat\psi_j$ formulas)
takes time at most polynomial in  the size of $\psi$ and~$\widehat\psi$, which is at most a polynomial in
$\log(\epsilon^{-1})$.
\end{proof}

Lemma~\ref{lem:fclonelimeffTransitive} may have wider applications in the study
of approximate counting problems.  Often, approximation-preserving reductions
between counting problems are complicated to describe and difficult to analyse,
owing to the need to track error estimates.  Lemma~\ref{lem:fclonelimeffTransitive}
suggests breaking the reduction into smaller steps, and analysing each of them
independently.  This assumes, of course, that the reductions are \omegadefinable,
but that often seems to be the case in practice.

\section{Relational clones and nonnegative functions}
\label{sec:relclones}

A function $F\in \codomainall$ is \emph{Boolean}\footnote{Note that ``Boolean'' applies to the codomain here, 
not the domain. As noted earlier, all of the functions that we consider have Boolean domains. 
This usage of ``Boolean function'' is unfortunate, but is well established in the literature. 
When the range is not a subset of $\{0,1\}$ we emphasise this fact by referring to the function as a ``pseudo-Boolean'' function.}
if its range
is a subset of $\B$.
Then $F$ encodes a relation $R$ as follows: $\vecx$ is in the relation~$R$ iff $F(\vecx)=1$.
We will not distinguish between relations and the Boolean functions that define them.
Suppose that $\calR \subseteq \codomainall$ is a set of relations/Boolean functions.
A pp-formula over $\calR$ is an existentially quantified product of
atomic formulas (this is called an $\exists \mathrm{CNF}(\calR)$-formula in \cite{CKZ}).
More precisely, a pp-formula~$\psi$ over~$\calR$ in variables $V'=\{v_1,\ldots,v_{n+m}\}$
has the form
$$
\psi=\exists\, v_{n+1},\ldots,v_{n+m}\,\bigwedge_{j=1}^s\phi_j,
$$
where $\phi_j$ are all atomic formulas over~$\calR$ in the variables~$V'$.
As before, the variables $V=\{v_1,\ldots,v_n\}$ are called ``free'', and the others, $V'\setminus V$, are called ``bound''.
The formula~$\psi$ specifies a Boolean function $R_\psi:\dom^n\to\B$
in the following way.
$R_\psi(\vecx) = 1$ if there is a vector $\vecy\in\dom^m$
such that  $ \bigwedge_{j=1}^s R_{\phi_j}(\vecx,\vecy)$ evaluates to ``$1$'',
where $\vecx$ and $\vecy$ are assignments
$
\vecx:\{v_1,\ldots,v_n\}\to\dom$ and
 $\vecy:\{v_{n+1},\ldots,v_{n+m}\}\to\dom$;
$R_\psi(\vecx)=0$ otherwise.
We call the pp-formula an \emph{implementation} of  $R_\psi$.

A \emph{relational clone} (often called a ``co-clone'') is a set of  relations
containing the equality relation and closed under finite Cartesian products, projections, and identification of variables.
A \emph{basis} \cite{CKZ} for the relational clone~$I$ is a set~$\calR$ of Boolean relations
such that the relations in~$I$ are exactly the relations that  can be implemented with a pp-formula over~$\calR$.
Every relational clone has such a basis.

For every set $\calR$ of Boolean relations, let
$\rclone \calR$ denote the set of relations that can be represented by a pp-formula over~$\calR\cup\{\EQ\}$.
It is well-known that
if $R\in \rclone\calR$ then $\rclone{\calR \cup \{R\}} = \rclone{\calR}$ (This can be proved similarly
to the proof of Lemma~\ref{lem:fcloneTransitive}.)
Thus, $\rclone\calR$ is in fact a relational clone with basis~$\calR$.

 A basis~$\calR$ for a relational clone $\rclone\calR$ is called a ``plain basis''
 \cite[Definition 1]{CKZ} if every member of $\rclone\calR$ is definable by a CNF$(\calR)$-formula (a pp-formula over~$\calR$ with no $\exists$).

Pseudo-Boolean functions~\cite{BorosHammer} are defined on the Boolean domain $D=\B$, and have codomain $\ring=\R$, the real numbers. For $n\in\Nat$, denote by $\codomainreal_n$ the set of all functions $\B^n\to\R$, and denote the set of functions of all arities by $\codomainreal=\codomainreal_0\cup\codomainreal_1\cup\codomainreal_2\cup\cdots$.
Note that any tuple $\vecx\in\B^n$ is the indicator function of a subset of $[n]$. 
We write $|\vecx|$ for the cardinality of this set, i.e. $|\vecx|=\sum_{j=1}^n x_j$.

For most of this paper, we restrict attention to the codomain $\ring=\Rnonneg$ of nonnegative real numbers. Then $\codomainnneg_n$ is the set given by replacing $\R$ by $\Rnonneg$ in the definition of $\codomainreal_n$, and then $\codomainnneg$ is defined analogously to $\codomainreal$.
We will also need to consider the \emph{permissive} functions in $\codomainnneg$. These are functions  which are positive everywhere, so the codomain $\ring=\Rpos$, the positive real numbers. Thus $\codomainpos_n$ and $\codomainpos$ are given by replacing $\R$ by $\Rpos$ in the definitions of $\codomainreal_n$ and $\codomainreal$.

The advantage of working with the Boolean domain is (i)~it has a well-developed theory of relational clones, and (ii)~the concept of a log-supermodular function exists (see~\S4). As explained in the introduction, the advantage of working with nonnegative real numbers is that we disallow cancellation, and potentially obtain a more nuanced expressibility/complexity landscape.

Given a function $F\in \codomainnneg$, let $R_F$ be
the function corresponding to the relation underlying~$F$.
That is, $R_F(\vecx)=0$ if $F(\vecx)=0$ and
$R_F(\vecx)=1$ if $F(\vecx)>0$.
The following straightforward lemma will be useful.
\begin{lemma}\label{lem:underlyingrelation}
Suppose $\calF \subseteq \codomainnneg$.
Then $$ \rclone{\{R_F \mid F\in \calF\}}
 =    \{R_F \mid F \in \fclone{\calF}\}.$$
\end{lemma}
\begin{proof}
Let $\calF$ be a subset of $\codomainnneg$.
First, we must show that, for any
$R \in  \rclone{\{R_F \mid F\in \calF\}}$, $R$ is in
$\{R_F \mid F \in \fclone{\calF}\}$.

Let $\psi$ be the pp-formula over
$\{R_F \mid F\in \calF\}\cup\{\EQ\}$ that is used to represent~$R$.
Write $\psi$ as
$$
\psi=\exists\,v_{n+1},\ldots,v_{n+m}\,\bigwedge_{j=1}^s
R_{F_j}(v_{i(j,1)},\ldots,v_{i(j,a_j)}),
$$
where  $F_j$ is an arity-$a_j$ function in~$\calF\cup\{\EQ\}$, and the index
function $i(\cdot,\cdot)$
picks out an index in the range $[1,n+m]$,
and hence a variable from
$V'=\{v_1,\ldots,v_{n+m}\}$. Let $\psi'$ be the \ppformula{} over
$\calF\cup\{\EQ\}$ given by
 $$
\psi'=\sum_{v_{n+1},\ldots,v_{n+m}}\,\prod_{j=1}^s
 F_j(v_{i(j,1)},\ldots,v_{i(j,a_j)}).
$$
Let $F' = F_{\psi'}$. Note that $F'\in \fclone{\calF}$ and that $R_{F'} = R$.

By reversing this construction, we can show that, for any
$R \in  \{R_F \mid F \in \fclone{\calF}\}$, $R$ is in
$\rclone{\{R_F \mid F\in \calF\}}$.
\end{proof}

\section{Log-supermodular functions}\label{sec:lsm}

A function $F \in \codomainnneg_n$
is \emph{log-supermodular} (lsm) if  $F(\vecxuy)F(\vecxny)\geq F(\vecx)F(\vecy)$ for all $\vecx,\vecy\in\B^n$.  The terminology is justified by the observation that $F\in\codomainpos_n$ is lsm if and only if $f=\log F$ is \emph{supermodular}. We denote by $\LSM\subset\codomainnneg$ the class of all lsm functions.
The second part of our main result (Theorem~\ref{thm:main}) says that, in terms of expressivity, everything of interest takes place within the class~$\LSM$.
Consequently, in \S\ref{sec:cloneLSMk}, we will investigate the internal structure of $\LSM$.

Note that $\codomainbool_0,\codomainbool_1\subset\LSM$, since log-supermodularity is trivial for nullary or unary functions, and hence the class $\LSM$ is conservative.
And it fits naturally into our study of expressibility
because of the following closure property:
functions that are \omegadef{} from lsm functions are lsm.
The non-trivial step in showing this is encapsulated in the following lemma.
It is a special case of the Ahlswede-Daykin ``four functions'' theorem~\cite{AD}. However~\cite{AD} is a much stronger result than is required, so we give an easier proof, using an argument similar to the base case of the induction in~\cite{AD}.

\begin{lemma}\label{lem:closedUnderProjection}
If $G \in \codomainnneg_{n+m}$, let
$G' \in \codomainnneg_n$ be defined by $G'(\vecx)=\sum_{\vecz\in\B^m}G(\vecx,\vecz)$.
Then $G\in\LSM$ implies $G'\in\LSM$.
\end{lemma}
\begin{proof}
By symmetry and induction, it suffices to consider summation on the last variable.
Thus, let $(\vecx,x_{n+1}),\hspace{1pt}(\vecy,y_{n+1})\in\B^{n+1}$, and let
\[\alpha_z=G(\vecx,z),\ \beta_z=G(\vecy,z),\ \gamma_z=G(\vecxuy,z),\  \delta_z=G(\vecxny,z)\quad(z\in\B).\]
Then we must show that $G\in\LSM$ implies
\begin{equation*}
(\alpha_0+\alpha_1)(\beta_0+\beta_1) \,\leq\, (\gamma_0+\gamma_1)(\delta_0+\delta_1),
\end{equation*}
which expands to
\begin{equation}
\alpha_0\beta_0+\alpha_1\beta_0+\alpha_0\beta_1+\alpha_1\beta_1\,\leq\, \gamma_0\delta_0+\gamma_0\delta_1+\gamma_1 \delta_0+\gamma_1\delta_1.\label{sumlsm:eq0}
\end{equation}

Since $G\in\LSM$, we have the following four inequalities,\\[-0.5\baselineskip]
\begin{minipage}{0.25\linewidth}
\begin{equation}
\alpha_0\beta_0 \leq \gamma_0\delta_0 \label{sumlsm:eq1},
\end{equation}
\end{minipage}
\begin{minipage}{0.25\linewidth}
\begin{equation}
\alpha_0\beta_1 \leq \gamma_1\delta_0  \label{sumlsm:eq2},
\end{equation}
\end{minipage}
\begin{minipage}{0.25\linewidth}
\begin{equation}
\alpha_1\beta_0 \leq \gamma_1\delta_0 \label{sumlsm:eq3},
\end{equation}
\end{minipage}
\begin{minipage}{0.25\linewidth}
\begin{equation}
\alpha_1\beta_1 \leq \gamma_1\delta_1  \label{sumlsm:eq4}.
\end{equation}
\end{minipage}\vspace{6pt}

We will complete the proof by showing that \eqref{sumlsm:eq1} to \eqref{sumlsm:eq4} imply \eqref{sumlsm:eq0} for arbitrary nonnegative real numbers $\alpha_z,\beta_z,\gamma_z,\delta_z$ $(z\in\B)$.

Using \eqref{sumlsm:eq1} and \eqref{sumlsm:eq4}, it follows that \eqref{sumlsm:eq0} is implied by
\begin{equation}\label{sumlsm:eq5}
\alpha_1\beta_0+\alpha_0\beta_1\,\leq\, \gamma_0\delta_1+\gamma_1\delta_0.
\end{equation}
Observe that $\gamma_1\delta_0=0$ implies \eqref{sumlsm:eq5}, since the left side is zero by \eqref{sumlsm:eq2} and \eqref{sumlsm:eq3}. Thus we may assume $\gamma_1\delta_0>0$.

Now, using \eqref{sumlsm:eq1} and \eqref{sumlsm:eq4} again, \eqref{sumlsm:eq5} is implied by
\begin{equation}\label{sumlsm:eq6}
\alpha_1\beta_0+\alpha_0\beta_1\,\leq\, \frac{(\alpha_0\beta_0)(\alpha_1\beta_1)}{\delta_0\gamma_1}+\gamma_1\delta_0\,=\, \frac{(\alpha_1\beta_0)(\alpha_0\beta_1)}{\gamma_1\delta_0}+\gamma_1\delta_0.
\end{equation}
Now\nobreakspace \textup {(\ref {sumlsm:eq6})} can be rewritten as
\begin{equation*}
0\,\leq\,(\gamma_1\delta_0-\alpha_1\beta_0)(\gamma_1\delta_0-\alpha_0\beta_1),
\end{equation*}
which is implied by \eqref{sumlsm:eq2} and \eqref{sumlsm:eq3}.
\end{proof}

\begin{lemma}\label{lem:lsmClosure}
If $\calF\subseteq\LSM$ then $\fclonelim\calF\subseteq\LSM$.
\end{lemma}

\begin{proof}
We just need to show that each level in the definition of \omegadef{} function preserves lsm:
first that  every
atomic formula over~$\calF\cup\{\EQ\}$ defines an lsm function,
then that a product of lsm functions is lsm,
then that a summation of an lsm function is lsm,
and finally that a limit of lsm functions is lsm.
As we shall see below,
only the third step is non-trivial, and it is covered by Lemma~\ref{lem:closedUnderProjection}.

First, note that the $\EQ$ is lsm, so every function in $\calF \cup\{\EQ\}$ is lsm,
An atomic formula $\phi=G(v_{i_1},\ldots,v_{i_a})$ defines
a function
$F_\phi(\vecx)=G(x_{i_1},\allowbreak\ldots,x_{i_a})$ which is lsm:
\begin{align*}
F_\phi(\vecxuy)F_\phi(\vecxny)&=
G(x_{i_1}\vee y_{i_1},\ldots,x_{i_a}\vee y_{i_a})\,G(x_{i_1}\wedge y_{i_1},\ldots,x_{i_a}\wedge y_{i_a})\\
&\geq G(x_{i_1},\ldots,x_{i_a})\,G(y_{i_1},\ldots,y_{i_a})\\
&=F_\phi(\vecx)F_\phi(\vecy).
\end{align*}
Note that we do not need to assume that $i_1,\ldots,i_a$ are all distinct.

It is immediate that the product of two lsm functions (and hence the product of
any number) is lsm.  Thus the product $\prod_{j=1}^sF_{\phi_j}$ appearing
in~\eqref{eq:ppfunctionDef} is lsm.  Then, by Lemma~\ref{lem:closedUnderProjection},
the \ppdef{} function~$F_\psi$ in~\eqref{eq:ppfunctionDef} is lsm.

Finally,
we will show that
any function that is approximated by lsm functions is lsm. Suppose that a function $F\in\codomainnneg_n$ has the property that,
for every $\epsilon>0$, there
is an arity-$n$ lsm function~$\Fhat$
satisfying
$$\|\Fhat-F\|_\infty=\max_{\vecx\in\B^a}|\Fhat(\vecx)-F(\vecx)|<\epsilon.$$
We wish to show that $F$ is lsm.
Let $\maxF = \max_{\vecx} F(\vecx)$.
Suppose for contradiction that $F$ is not lsm,
so there is a $\delta>0$ and
$\vecx, \vecy \in \B^n$
such that
$$F(\vecxuy)F(\vecxny)\leq F(\vecx)F(\vecy) - \delta.$$
Let $\epsilon>0$ be sufficiently small
that  $\epsilon \max(\Fmax,1)$ is tiny compared to $\delta$.
Then
\begin{align*}
\hatF(\vecxuy)\hatF(\vecxny)
&\leq (F(\vecxuy) + \epsilon)(F(\vecxny) + \epsilon) \\
&\leq F(\vecxuy)F(\vecxny) + 2 \epsilon  \maxF + \epsilon^2 \\
&\leq  F(\vecx)F(\vecy) - \delta + 2 \epsilon  \maxF + \epsilon^2\\
&\leq (\hatF(\vecx) + \epsilon)(\hatF(\vecy)+\epsilon) - \delta + 2 \epsilon  \maxF + \epsilon^2\\
&\leq \hatF(\vecx) \hatF(\vecy) + 2 \epsilon (\maxF + \epsilon) - \delta + 2 \epsilon  \maxF + 2\epsilon^2\\
&< \hatF(\vecx) \hatF(\vecy) ,
\end{align*}
so $\hatF$ is not lsm, giving a contradiction.
\end{proof}

An important example of an lsm function is the 0,1-function ``implies'',
$$
\IMP(x,y)=\begin{cases}0,&\text{if $(x,y)=(1,0)$};\\1,&\text{otherwise}.\end{cases}
$$
We also think of this as a binary relation $\IMP=\{(0,0),(0,1),(1,1)\}$.
Complexity-theoretic issues will be treated in detail in~\S\ref{sec:complexity}.
However, it may be helpful to give a pointer here to the importance of IMP in the
study of approximate counting problems.

The problem $\BIS$ is that of counting
independent sets in a bipartite graph.  Dyer et al.~\cite{DGGJ} exhibited a class of
counting problems, including $\BIS$, which are interreducible via approximation-preserving
reductions.  Further natural problems have been shown to lie in this class,
providing compelling evidence that is of intermediate complexity
between counting problems that are tractable (admit a polynomial-time approximation algorithm)
and those that are $\NP$-hard to approximate.
We will see in due course (Theorem~\ref{corthm:main} and Proposition~\ref{lem:weights})
that $\BIS$ and $\nCSP(\IMP)$ are interreducible
via approximation-preserving reductions, and hence are of equivalent
difficulty.

We know from Lemma~\ref{lem:lsmClosure} that $\fclonelim{\IMP,\codomainnneg_1}\subseteq\LSM$, and one might ask whether this inclusion is strict. We will address this question in \S\ref{sec:cloneLSMk}.

\section{Pinnings and modular functions}\label{sec:pinnings}

Let $\unaryzero$ be the unary function with $\unaryzero(0)=1$ and $\unaryzero(1)=0$ and let $\unaryone$ be the unary function with $\unaryone(0)=0$ and $\unaryone(1)=1$.

Let $S\subseteq[n]$, let $\vecx'=(x_j)_{j\notin S}$, and $\vecx''=(x_j)_{j\in S}$, and partition $\vecx\in\B^n$ as $(\vecx';\vecx'')$. Then, if $F\in\codomainbool_{n}$, the function $F(\vecx';\vecc)$ given by setting  $x''_j=c_j$ for constants $c_j\in\B$ $(j\in S)$ is a \emph{pinning} of $F$. Note that we allow the empty pinning $S=\emptyset$, which is $F$ itself, and the pinning of all variables $S=[n]$, which is a nullary function.

Clearly, every pinning of~$F$ is in $\fclone{F,\delta_0,\delta_1}$, since a constant~$c\in\B$ can be implemented using either~$\unaryzero$ or~$\unaryone$, i.e. we add $\delta_c(x_i)$ to the constraint set. We will use the notation $i\gets c$ to indicate that the $i$th variable has been pinned to $c$.

If $n\geq 2$ then a \emph{2-pinning} of a function $F\in \codomainbool_n$
is a function $F(x_i,x_j;\vecc)$ which pins \emph{all but} 2 of the variables. Thus $[n]\setminus S=\{i,j\}$, where $i$ and $j$ are distinct indices in $[n]$, and $\vecc\in\B^{n-2}$. Clearly, every $2$-pinning of $F$ is in $\fclone{F,\codomaineff_1}$, since it is a pinning of $F$.

We say that a function $F\in \codomainnneg_n$
is log-modular if $F(\vecxuy)F(\vecxny) =  F(\vecx)F(\vecy)$
for all $\vecx,\vecy\in\B^n$.

It is a fact that $\LSM$ and the class of log-modular functions are closed under pinning. This is a consequence of the following lemma about $2$-pinnings of  lsm and log-modular functions. It is due, in essence, to Topkis~\cite{Topkis}, but we provide a short proof for completeness.
\begin{lemma}[Topkis]\label{lem:Topkis}
A function $F\in\codomainpos$ is lsm iff every $2$-pinning is lsm, and is log-modular iff every $2$-pinning is log-modular.
\end{lemma}
\begin{proof}
The necessity of the $2$-pinning condition is obvious, but we must prove sufficiency. We need only show $F(\vecx)F(\vecy)\leq F(\vecxuy)F(\vecxny)$ for $\vecx,\vecy$ such that $x_i\neq y_i$ $(i\in[n])$. All other cases follow from this.
Note that log-supermodularity is preserved under arbitrary permutation of variables. Thus, if $0^r$, $1^r$ denote $r$-tuples of $0$'s and $1$'s respectively, we must show that, for all $r, s>0$ with $r+s=n$,
\begin{equation}
F(0^r,1^s)F(1^r,0^s) \leq F(1^r,1^s)F(0^r,0^s).\label{eq:pinning1}
\end{equation}
We will prove this by induction, assuming it is true for all $r',s'>0$ such
that $r'+s'<n$. The base case, $r'=s'=1$, is the 2-pinning assumption.
If $r>1$, then we have
\begin{align}
F(0^r,1^s)F(1^{r-1},0^{s+1}) &\leq F(1^{r-1},0,1^s)F(0^r,0^s),\label{eq:pinning2}
\intertext{by induction, after pinning the $r$th position to $0$,}
F(1^{r-1},0,1^s)F(1^r,0^s) &\leq F(1^r,1^s)F(1^{r-1},0^{s+1})\label{eq:pinning3}
\end{align}
by induction, after pinning the first $r-1>0$ positions to $1$.\vspace{1ex}

Now, multiplying \eqref{eq:pinning2} and \eqref{eq:pinning3} gives \eqref{eq:pinning1} after cancellation, which is valid since $F$ is permissive. If $r=1$, we do not have the induction giving \eqref{eq:pinning3}, so we use instead
\begin{align}
F(1,0^s)F(0,1,0^{s-1}) &\leq F(0,1^s) F(1,1,0^{s-1})\label{eq:pinning4},\\
\intertext{by induction, using the base case after pinning the last $s-1$ positions to $0$,}
F(1,1,0^{s-1})F(0,1^s)
&\leq F(1,1^s)F(0,1,0^{s-1})\label{eq:pinning5}.
\end{align}
by induction, after pinning the second position to $1$.\vspace{1ex}

Now, multiplying \eqref{eq:pinning4} and \eqref{eq:pinning5} gives \eqref{eq:pinning1} after cancellation,
completing the proof for log-supermodularity.
The proof for log-modularity is identical, except that every ``$\leq$'' must be replaced by ``$=$''.
\end{proof}
\begin{remark}
    We have proved Lemma~\ref{lem:Topkis} only for permissive functions because, in fact, it is false more generally. Consider, for example, the function $F\in\codomainbool_4$ such that $F(1,1,0,0)=F(0,0,1,1)=1$, $F(\vecx)=0$ otherwise. It is easy to see that all the 2-pinnings $F'$ of $F$ have $F'(x,y)>0$ for at most one $(x,y)\in\B^2$. It follows that every 2-pinning of $F$ is log-modular. But $F$ is not even lsm, since $F(1,1,0,0)F(0,0,1,1) > F(1,1,1,1)F(0,0,0,0)$.
\end{remark}

\section{Computable real numbers}\label{sec:reals}

Since we want to be able to derive computational results,
we will now focus attention on functions whose co-domain
is restricted to efficiently-computable real numbers. We will say that
a real number is polynomial-time computable if the $n$ most significant bits
of its binary expansion can be computed in time polynomial in $n$.
This is essentially the definition given in~\cite{KoFri82}.
Let $\Reff$ denote the set of \emph{nonnegative}
real numbers that are polynomial-time computable.
For $n\in\Nat$, denote by $\codomaineff_n$  the set of all functions
$\B^n\to\Reff$;
also denote by $\codomaineff=\codomaineff_0\cup\codomaineff_1\cup\codomaineff_2\cup\cdots$ the set of functions
of all arities.

\begin{remark}
As we have defined them, it is known that the polynomial-time computable numbers form a field~\cite{KoFri82}, and hence a subsemiring $\ring$ of $\Cnum$, as we  require. Thus, in our definitions, there is no difficulty with \ppsdefinability. However, there could be a problem with \omegadefinability, since the limit of a sequence of polynomial-time computable reals may not be polynomial-time computable. Polynomial-time computability is ensured only by placing restrictions on speed of convergence. See~\cite{KoFri82,Tora87} for details. However, observe that our definition of efficient \omegadefinability\  avoids this difficulty entirely, by insisting that the limit of a sequence of reals will be permitted only if the limit is itself polynomial-time computable.
\end{remark}

\begin{remark}
The polynomial-time computable real numbers are a proper subclass of the \emph{efficiently approximable} real numbers, defined in~\cite{GoldbergJ10}. (This fact can be deduced from~\cite{KoFri82}.) We have made this restriction since it results in a more uniform treatment of limits when we discuss efficient \omegadefinability{} for functions in $\codomaineff$.
\end{remark}

\section{Binary functions}\label{sec:binary}
We begin the study of the conservative case in the simplest nontrivial situation. We  consider the functional clones $\fclonelimeff{F,\codomaineff_1}$, where $F$ is a single binary function.

Recall that $\EQ$ is the binary relation $\EQ = \{(0,0),(1,1)\}$.
(We used the name ``$\EQ$'' to denote the equivalent binary function, but
it will do no harm to use the same symbol for the relation and the function.)
Denote by $\OR$, $\NEQ$,   and $\NAND$
the binary relations  $\OR=\{(0,1),(1,0),(1,1)\}$, $\NEQ=\{(0,1),(1,0)\}$,
and $\NAND  = \{(0,0),(0,1),(1,0)\}$.
When we write a function
$F\in\codomainbool_2$, we will identify the arguments by writing $F(x_1,x_2)$.
We may represent $F$ by a $2\times 2$ matrix
\[  M(F)\ =\ \begin{bmatrix} F(0,0) & F(0,1)\\F(1,0) & F(1,1) \end{bmatrix}
\ =\
\begin{bmatrix} f_{00} & f_{01}\\f_{10} & f_{11} \end{bmatrix},\]
say, with rows indexed by $x_1\in\B$ and columns by $x_2\in\B$. We will assume $f_{01}\geq f_{10}$, since otherwise we may consider the function $F^T$, such that $F^T(x_1,x_2)=F(x_2,x_1)$, represented by the matrix $M(F)^T$. Clearly $\fclone{F^T}=\fclone{F}$.

If $U$ is a unary function, we will write $U=(U(0),U(1))=(u_0,u_1)$, say. Then we have
\[ M\big(U(x_1)F(x_1,x_2)\big)=\begin{bmatrix} u_0f_{00} & u_0f_{01}\\u_1f_{10} & u_1f_{11} \end{bmatrix},\quad  M\big(U(x_2)F(x_1,x_2)\big) = \begin{bmatrix} u_0f_{00} & u_1f_{01}\\u_0f_{10} & u_1f_{11} \end{bmatrix},\]
where both $U(x_1)F(x_1,x_2)$ and $U(x_2)F(x_1,x_2)$ are clearly in $\fclone{F,U}$.

If $F_1,F_2\in\codomainbool_2$, then $M(F_1)M(F_2)=M(F)$, where $F\in\fclone{F_1,F_2}$ is such that
\[ F(x_1,x_2)\ =\ \sum_{y=0}^1F_1(x_1,y)F_2(y,x_2).\]
\begin{lemma}\label{lem:binarycases}
Let $F\in\codomaineff_2$. Assuming $f_{01}\geq f_{10}$,
\begin{enumerate}[topsep=5pt,itemsep=0pt,label=(\roman*)]
  \item if $f_{00}f_{11}=f_{01}f_{10}$, then $\fclonelimeff{F,\codomaineff_1}=\fclonelimeff{\codomaineff_1}$;\label{case1}
  \item if $f_{01},f_{10}=0$ and $f_{00},f_{11}>0$, then $\fclonelimeff{F,\codomaineff_1}=\fclonelimeff{\codomaineff_1}$;\label{case2}
  \item if $f_{00},f_{11}=0$ and $f_{01},f_{10}>0$, then $\fclonelimeff{F,\codomaineff_1}=\fclonelimeff{\NEQ,\codomaineff_1}$;\label{case3}
  \item if $f_{00},f_{01},f_{11}>0$ and $f_{00}f_{11}>f_{01}f_{10}$, then $\fclonelimeff{F,\codomaineff_1}=\fclonelimeff{\IMP,\codomaineff_1}$;\label{case4}
  \item otherwise, $\fclonelimeff{F,\codomaineff_1}=\fclonelimeff{\OR,\codomaineff_1}$.\label{case5}
\end{enumerate}
The non-efficient version --- with $\codomainbool_1,\codomainbool_2$ replacing
$\codomaineff_1,\codomaineff_2$, and $\fclonelim\cdot$ replacing $\fclonelimeff\cdot$ --- also holds.
\end{lemma}
\begin{proof}To prove $\fclonelimeff{F_1,\codomaineff_1}=\fclonelimeff{F_2,\codomaineff_1}$, it suffices to show that $F_2\in\fclonelimeff{F_1,\codomaineff_1}$ and $F_1\in\fclonelimeff{F_2,\codomaineff_1}$. We will verify this in each of the five cases.
\begin{enumerate}[topsep=5pt,itemsep=5pt,label=(\roman*)]
  \item Suppose $f_{00}f_{11}=f_{01}f_{10}$. If $f_{00},f_{01}=0$, then
  $$F(x_1,x_2)=U_1(x_1)U_2(x_2)$$
      with
  $U_1=(0,1)$ and $U_2=(f_{10},f_{11})$.
    Similarly if $f_{00},f_{10}=0$, $f_{01},f_{11}=0$, or $f_{10},f_{11}=0$. In the remaining case $f_{00},f_{01},f_{10},f_{11}>0$. Then choose $U_1=(1,f_{10}/f_{00})$, $U_2=(f_{00},f_{01})$. In all cases $F\in\fclone{U_1,U_2}$,
      so $\fclonelimeff{F,\codomaineff_1}=\fclonelimeff{\codomaineff_1}$.
  \item if $f_{01},f_{10}=0$ and $f_{00},f_{11}>0$, then $F(x_1,x_2)=U(x_1)\EQ(x_1,x_2)$, where $U=(f_{00},f_{11})$, so $F\in\fclone{U}$.
Hence  $\fclonelimeff{F,\codomaineff_1}=\fclonelimeff{\codomaineff_1}$.
  \item If $f_{00},f_{11}=0$ and $f_{01},f_{10}>0$, then $F(x_1,x_2)=U(x_1)\NEQ(x_1,x_2)$, where $U=(f_{01},f_{10})$, so $F\in\fclone{\NEQ,U}$. Similarly $\NEQ(x_1,x_2)=U'(x_1)F(x_1,x_2)$, where $U'=(1/f_{01},1/f_{10})$, so $\NEQ \in\fclone{F,U'}$. So $\fclonelimeff{F,\codomaineff_1}=\fclonelimeff{\NEQ,\codomaineff_1}$.
 \item If $f_{00},f_{01},f_{11}>0$, $f_{00}f_{11}>f_{01}f_{10}$,
  we can apply unary weights $U_1,U_2$, where $U_1=(1/f_{00},f_{01}/f_{00}f_{11})$, $U_2=(1,f_{00}/f_{01})$, to implement $\IMP_\alpha(x_1,x_2)=U_1(x_1)U_2(x_2)F(x_1,x_2)$, where
\[  M(\IMP_\alpha)\ =\ \begin{bmatrix} 1 & 1\\  \alpha & 1 \end{bmatrix},
\]
where $\alpha=f_{01}f_{10}/f_{00}f_{11}<1$. Then we have $\IMP_\alpha\in\fclone{F,U_1,U_2}$. Note that $\IMP_0=\IMP$. If $\alpha>0$, consider the function $\IMP_\alpha^k$, with matrix
\[  M(\IMP_\alpha^k)\ =\ \begin{bmatrix} 1 & 1\\  \alpha^k & 1 \end{bmatrix}.
\]
Now $\IMP_\alpha^k$ can be implemented
as $\IMP_\alpha^k(x_1,x_2)=U_1^k(x_1)U_2^k(x_2)F^k(x_1,x_2)$, by taking  $k$ copies of $U_1$, $U_2$ and $F$.
Since $\alpha<1$, we see that $\lim_{k\to\infty}\IMP_\alpha^k=\IMP_0=\IMP$. Moreover, the limit is efficient, since $\|\IMP-\IMP_\alpha^k\|<\epsilon$ if $k=O(\log\epsilon^{-1})$, and so an $\epsilon$-approximation to $\IMP$ can be computed in $O(\log\epsilon^{-1})$ time. Hence $\IMP\in\fclonelimeff{F,\codomaineff_1}$.

Note that ``powering'' limits like that used here will be employed below without further discussion of their efficiency.

Conversely, from $\IMP$, we first implement $\IMP_\alpha$. If $\alpha=0$, we do nothing. Otherwise, we use unary weights $U_1, U_2$ such that $U_1=(1/\alpha-1,1)$, $U_2=(\alpha,1)$, to implement $F_1(x_1,x_2)=U_1(x_1)U_2(x_2)\IMP(x_2,x_1)$, where
\[ M(F_1)\ =\ \begin{bmatrix} 1-\alpha & 0 \\  \alpha & 1 \end{bmatrix}. \]
Then $M(\IMP_\alpha)=M(\IMP)M(F_1)$, so $\IMP_\alpha\in\fclone{\IMP,U_1,U_2}$.
Now we can recover
$F(x_1,x_2)=U_3(x_1)U_4(x_2)\IMP_\alpha(x_1,x_2)$,
where $U_3=(f_{00},f_{00}f_{11}/f_{01})$, $U_4=(1,f_{01}/f_{00})$, so we have $F\in\fclone{\IMP,U_1,U_2,U_3,U_4}$.
Hence $\fclonelimeff{F,\codomaineff_1}=\fclonelimeff{\IMP,\codomaineff_1}$.
\item The remaining cases are
(a)
$f_{01},f_{10},f_{11}>0$, $f_{00}f_{11}<f_{01}f_{10}$ and
(b) $f_{00}$, $f_{01}$, $f_{10}>0$, $f_{11}=0$.

First, we deal with part~(a): If $f_{01},f_{10},f_{11}>0$ and $f_{00}f_{11}<f_{01}f_{10}$, we apply unary weights $U_1,U_2$, where $U_1=(f_{11}/f_{01},1)$, $U_2=(1/f_{10},1/f_{11})$, to implement $\OR_\alpha(x_1,x_2)=U_1(x_1)U_2(x_2)F(x_1,x_2)$, where $\alpha=f_{00}f_{11}/f_{01}f_{10}<1$, and
  \[ M(\OR_\alpha)=\begin{bmatrix} \alpha & 1\\  1 & 1 \end{bmatrix}\]
   If $\alpha=0$, $\OR_0=\OR$, so we have $\OR\in\fclone{F,\codomaineff_1}$. Otherwise $\lim_{k\to\infty}\OR_\alpha^k=\OR_0=\OR$, so we have $\OR\in\fclonelimeff{F,\codomaineff_1}$.

   Conversely, from $\OR$, we first
       express
    $\NEQ$. Use the unary function $U=(2,\half)$ to implement $F_1=U(x_1)U(x_2)\OR(x_1,x_2)$, where
\[ M(F_1)=\begin{bmatrix} 0 & 1\\  1 & \nicefrac{1}{4} \end{bmatrix}.\]
 Then $\lim_{k\to\infty} F_1^k=\NEQ$, so $\NEQ\in\fclonelimeff{\OR,\codomaineff_1}$. Now we observe that $M(\IMP)=M(\NEQ)M(\OR)$, so $\IMP\in\fclonelimeff{\OR,\codomaineff_1}$.
Now we have $\IMP_\alpha\in\fclonelimeff{\OR,\codomaineff_1}$, as in~\ref{case4} above. Then $M(\OR_\alpha)=M(\NEQ)M(\IMP_\alpha)$, so $\OR_\alpha\in\fclonelimeff{\OR,\codomaineff_1}$. Now we can reverse the transformation from $F$ to $\OR_\alpha$ to recover $F$.

Now, we consider part~(b):
If $f_{00},f_{01},f_{10}>0$ and $f_{11}=0$,
  we apply unary weights $U_1,U_2$, where $U_1=(1/f_{00},1/f_{10})$, $U_2=(1,f_{00}/f_{01})$, to implement $\NAND(x_1,x_2)=U_1(x_1)U_2(x_2)F(x_1,x_2)$, where
  \[ M(\NAND)=\begin{bmatrix} 1 & 1\\  1 & 0 \end{bmatrix},\]
   so we have $\NAND\in\fclone{F,\codomaineff_1}$. We now use unary weight $U=(\half,2)$ to implement $F_1(x_1)=U(x_1)U(x_1)\NAND(x_1,x_2)$ with
   \[ M(F_1)=\begin{bmatrix} \nicefrac{1}{4} & 1\\  1 & 0 \end{bmatrix}.\]
  Then we have $\lim_{k\to\infty} F_1^k=\NEQ$, so again $\NEQ\in\fclonelimeff{F,\codomaineff_1}$. Then we observe that $M(\OR)=M(\NEQ)M(\NAND)M(\NEQ)$, so $\OR\in\fclonelimeff{F,\codomaineff_1}$.

Conversely, from $\OR$, we have $\NEQ\in\fclonelimeff{\OR,\codomaineff_1}$ from the above. Then we have $M(\NAND)=M(\NEQ)M(\OR)M(\NEQ)$,
so $\NAND\in\fclonelimeff{\OR,\codomaineff_1}$. Now we reverse the transformation above from $F$ to $\NAND$ to recover $F$. Thus $F\in\fclonelimeff{\OR,\codomaineff_1}$.\qedhere
\end{enumerate}
\end{proof}
\begin{remark}
From Lemma~\ref{lem:binarycases}, we see that $\IMP$ does not really occupy a special position in $\fclonelimeff{\IMP,\codomaineff_1}$, in the sense that there are other
functions
$F$ with $\fclonelimeff{F,\codomaineff_1}=\fclonelimeff{\IMP,\codomaineff_1}$.
Similarly,
$\OR$ does not occupy a special position
in $\fclonelimeff{\OR,\codomaineff_1}$. Nevertheless, it is useful to label the classes this way, and we will do so.
\end{remark}
\begin{remark}\label{rem:inclusions}
From the proof of Lemma~\ref{lem:binarycases}, we have the following inclusions between the four classes involved.
\[\fclonelimeff{\codomaineff_1}\,
\subseteq \begin{array}{c}\fclonelimeff{\NEQ,\codomaineff_1} \\[0.5ex] \fclonelimeff{\IMP,\codomaineff_1}\end{array}
\subseteq\fclonelimeff{\OR,\codomaineff_1}.\]
In fact, $\fclonelimeff{\NEQ,\codomaineff_1}$ and $\fclonelimeff{\IMP,\codomaineff_1}$ are incomparable,
and hence all the inclusions are actually strict.
 For one non-inclusion,
note the clone $\fclonelimeff{\IMP,\codomaineff_1}$ contains only lsm functions, and hence does
not contain~$\NEQ$.  For the other, we claim that any
binary function in the clone $\fclonelimeff{\NEQ,\codomaineff_1}$
has one of three forms, $U_1(x)U_2(y)$, $U(x)\EQ(x,y)$ or $U(x)\NEQ(x,y)$,
and then observe that $\IMP$ matches none of these.
The claim is a special case of a more general one, namely that any
function in $\fclonelimeff{\NEQ,\codomaineff_1}$ is of the form~$F_\phi$, where $\phi$ is a product of atomic formulas
involving only unary functions and $\NEQ$.
To show this, we need only consider summing over a single variable. By induction, assume we have a \ppformula\ of the form
$\sum_{y\in\B} F_\phi(\vecx,\vecz) F_\psi(\vecx,y)$, where $F_\psi$ is a product of atomic formulas involving~$y$
(and certain other variables $\vecx$), and $F_\phi$ is a product not involving~$y$.  Then
\begin{align*}
  \sum_{y\in\B} F_\phi(\vecx,\vecz) F_\psi(\vecx,y) &= F_\phi(\vecx,\vecz)\sum_{y\in\B} U(y)\prod_{i=1}^k\NEQ(x_i,y)\\
  &= F_\phi(\vecx,\vecz)\hspace{1pt}\bar{U}(x_1)\prod_{i=1}^k \EQ(x_i,x_1),
\end{align*}
where $\bar{U}(x_1)=U(1-x_1)$. The product of equalities can be removed by substituting $x_j$ $(j=2,\ldots,k$) by $x_1$ in $\phi(\vecx,\vecz)$, continuing the induction.

Finally, it is straightforward to show that this class is closed under limits. That is, the limit of a sequence of functions which are products of unary and $\NEQ$ functions must itself be of this form. To see this, note that every $k$-ary function in this class can be written as a product of $O(k^2)$ unary and $\NEQ$ functions. Then the conclusion follows by a standard compactness argument. The efficient version also follows.
\end{remark}

\section{The class $\fclonelimeff{\OR,\codomaineff_1}$}\label{sec:OR}
In the conservative case, we show that, somewhat surprisingly, the clone generated by the single binary function $\OR$ contains every function in $\codomaineff$.
\begin{lemma}\label{lem:OR}
$\fclonelimeff{\OR,\codomaineff_1}=\codomaineff$.
\end{lemma}

\begin{proof}
Suppose $F\in \codomaineff_n$. Suppose $x_1,\ldots,x_n$ are variables. For each $A\subseteq[n]$, let $\ind$ be the assignment to $x_1,\ldots,x_n$ in which $x_i=1$ if $i\in A$ and $x_i=0$ otherwise. Then we will use the notation $F(A)$ as shorthand for $F(\ind)$. Let $\calA = \{A : F(A)>0\}$, and let $\mu=\min_{A\in\calA}F(A)$. For any $A\subseteq [n]$, let $u_A\in\codomaineff_1$ be the function such that $u_A(0)=1$ and $u_A(1) = 2F(A)/\mu-1\geq 1$. Note that every function $u_A$ is in $\codomaineff_1$ and we have $\IMP,\NAND\in \fclonelimeff{\OR,\codomaineff_1}$, from the proof of Lemma~\ref{lem:binarycases}.

Our goal will be to show that there is a finite subset $S_F$ of $\{\IMP,\NAND\}\cup\codomaineff_1$ and a TM $M_{F,S_F}$
with the following property:
on input $\epsilon>0$, $M_{F,S_F}$ computes
an arity-$n$ \ppformula~$\psi$ over $S_F$
such that $\| F_\psi-F\|_\infty<\varepsilon$.
The running time of $M_{F,S_F}$ should be at most a polynomial in  $\log\epsilon^{-1}$.
To define $S_F$, we will use two unary functions~$U_1$ and~$U_2$
(both of which are actually constant functions).
We define these by  $U_1(0)=U_1(1)=\half$ and $U_2(0)=U_2(1)=\mu/2$.
Then $S_F = \{\IMP,\NAND,U_1,U_2\}\cup \bigcup_{A\in \calA} \{u_A\}$.

Let $V = \{v_1,\ldots,v_n\}$.
For $A\in \calA$, introduce a new variable~$z_A$.
Let $V'' = V \cup \{z_A \mid A\in \calA\}$.
Let
$$\psi_1 = \sum_{V''} \left(\prod_{A\in \calA} u_A(z_A)\right)\left(\prod_{i\in A}\IMP(z_A,x_i)\right)\left(\prod_{i\notin A}\NAND(z_A,x_i)\right).$$

For every $A\in \calA$
the assignment $\vecx=\ind$
can be extended in two ways (both with $z_A=0$ and with $z_A=1$) to satisfy
\begin{equation}
\label{eqstar}
\left(\prod_{i\in A}\IMP(z_A,x_i)\right)\left(\prod_{i\notin A}\NAND(z_A,x_i)\right)=1.
\end{equation}
Any other assignment $\vecx$ can be
extended in only one way ($z_A=0$) to satisfy (\ref{eqstar}).
So if $A\in \calA$
then
$$F_{\psi_1}(A) = (2F(A)/\mu-1) + 1 = 2F(A)/\mu.$$
On the other hand, if $A\notin\calA$
then
$$F_{\psi_1}(A) = 1.$$
We have shown that $F_{\psi_1}\in\fclone{S_F}$.
Let us now define
$$\psi_2 = \sum_{V''} \left(\prod_{A\in \calA}\prod_{i\in A}\IMP(z_A,x_i)\right)\left(\prod_{i\notin A}\NAND(z_A,x_i)\right).$$
As before, for every $A\in \calA$ the assignment $\vecx=\ind$
can be extended in two ways ($z_A=0$ and $z_A=1$) to satisfy
\eqref{eqstar}, and any other assignment $\vecx$ can be
extended in only one way ($z_A=0$) to satisfy it. So
$$F_{\psi_2}(A) = 2\quad(A\in \calA),\qquad F_{\psi_2}(A) = 1\quad(A\notin \calA).$$
Thus $F_{\psi_2}\in\fclone{ S_F}$. Now define $F_3$ by $F_3(A) = U_1(x_1)F_{\psi_2}(A)$, so $F_3\in\fclone{ S_F}$, where
$$F_3(A) = 1\quad(A\in \calA),\qquad F_3(A) = \half\quad(A\notin \calA).$$
Now $\lim_{k\to\infty} F_3^k=F_0$, where
$$F_0(A) = 1\quad(A\in \calA),\qquad F_0(A) = 0\quad(A\notin \calA),$$
and thus $F_0\in\fclonelimeff{S_F}$.

Note that $F_0=R_F$, the underlying relation of $F$. Now define $F_4=F_{\psi_1}F_0$, so that
$$F_4(A) = 2F(A)/\mu\quad(A\in \calA),\qquad F_4(A) = 0\quad(A\notin \calA),$$
Thus, by Lemma~\ref{lem:fclonelimeffTransitive},
$F_4\in\fclonelimeff{ S_F}$.  Now define $F_5$ by $F_5(A) = U_2(x_1)F_4(A)$, so $F_5\in\fclonelimeff{ S_F}$, where
$$F_5(A) = F(A)\quad(A\in \calA),\qquad F_5(A) = 0\quad(A\notin \calA).$$
Since $F_5=F$, the proof is complete.
\end{proof}

\section{The main theorem}\label{sec:main}
In this section we prove our main structural result, which characterises, in the conservative case, the clones generated by a single pseudo-Boolean function. Since it is known that any clone generated by a finite set of functions can be generated by a single function~\cite{BDGJJR12}, we have a characterisation of all
finitely generated
functional clones.

\begin{theorem}\label{thm:main}
Suppose $F\in\codomaineff$.
\begin{itemize}
\item If $F\notin\fclone{\NEQ,\codomaineff_1}$ then $\IMP\in\fclonelimeff{F,\codomaineff_1}$,
and hence $\fclonelimeff{\IMP,\codomaineff_1}\subseteq\fclonelimeff{F,\codomaineff_1}$
\item If, in addition, $F\notin\LSM$ then $\fclonelimeff{F,\codomaineff_1}=\codomaineff$.
\end{itemize}
The non-efficient version --- with $\codomainbool,\codomainbool_1$ replacing
$\codomaineff,\codomaineff_1$, and $\fclonelim\cdot$ replacing $\fclonelimeff\cdot$ --- also holds.
\end{theorem}

\begin{proof}
We start with the first part of the theorem. The aim is to show that
either  $\IMP\in\fclonelimeff{F,\codomaineff_1}$
or $F\in\fclone{\NEQ,\codomaineff_1}$.
Let $C$ be the relational clone $\rclone{R_F,  \unaryzero, \unaryone}$.
Since $\{R_F,\delta_0,\delta_1\} \subseteq \{R_{F'} \mid F' \in  \{F\} \cup \codomaineff_1  \}$,
$C \subseteq \rclone {R_{F'} \mid F' \in  \{F\} \cup \codomaineff_1}$,
so by Lemma~\ref{lem:underlyingrelation},
$C \subseteq \{ R_{F'} \mid F' \in \fclone{F,\codomaineff_1} \}$.

First, suppose $\IMP \in C$.
Then $\fclonelimeff{F,\codomaineff_1}$ contains a function $F'$ such that $R_{F'} = \IMP$.
The function~$F'$ falls into parts (iv) or (v)  of Lemma~\ref{lem:binarycases},
so by this lemma,
$\fclonelimeff{F,\codomaineff_1}$ is either $\fclonelimeff{\IMP,\codomaineff_1}$ or $\fclonelimeff{\OR,\codomaineff_1}$.
Either way,
$\fclonelimeff{F,\codomaineff_1}$
contains $\IMP$ (as noted in Remark~\ref{rem:inclusions}).
Similarly, if $\OR\in C$ or $\NAND\in C$
then $\IMP\in\fclonelimeff{F,\codomaineff_1}$.

We now consider the possibilities.
If $R_F$ is not affine, then Creignou, Khanna and Sudan \cite[Lemma 5.30]{CKS}
have shown that one of $\IMP$, $\OR$ and $\NAND$ is in~$C$.
This is also stated and proved as \cite[Lemma 15]{trichotomy}.
\begin{figure}[!t]
\centerline{\includegraphics{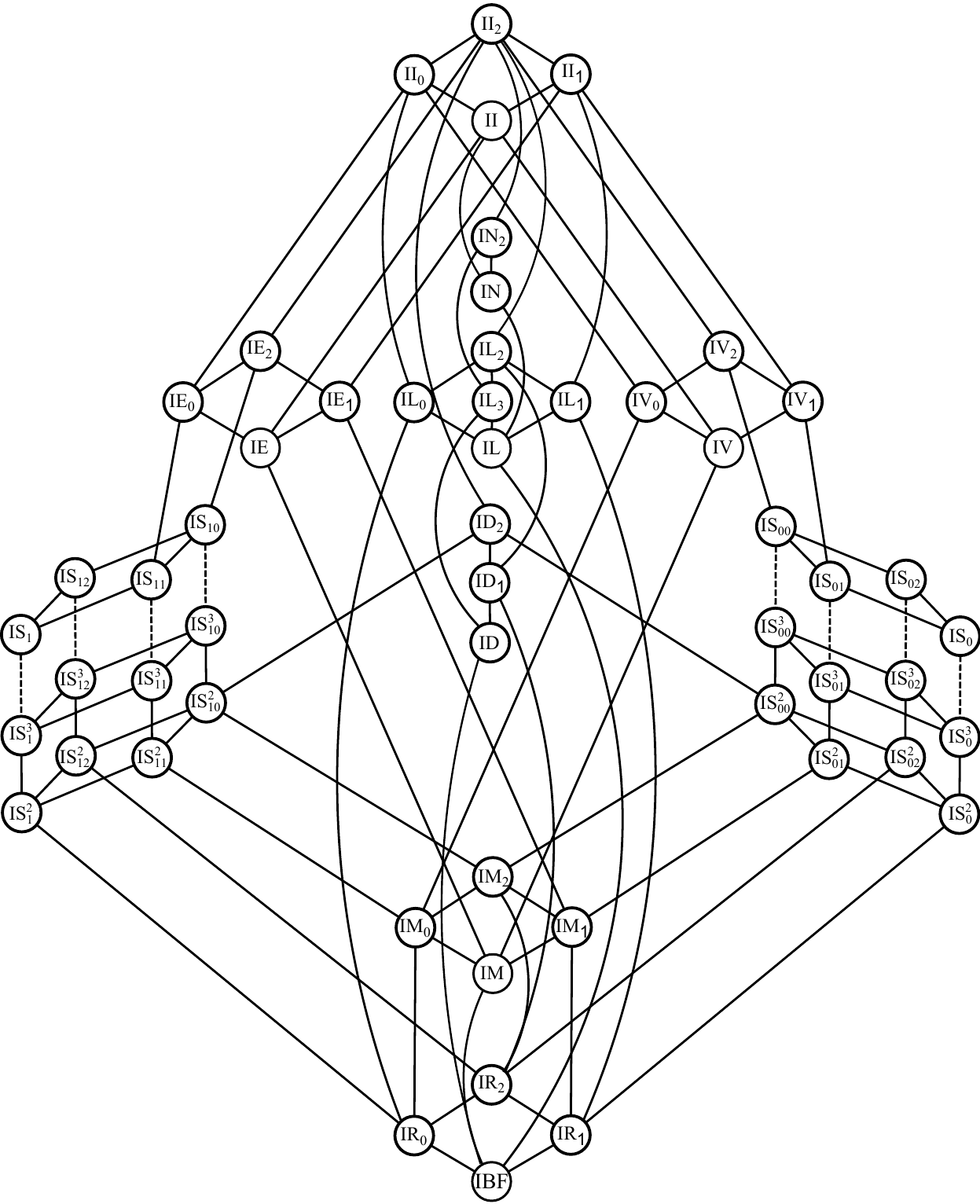}}
\caption{Post's lattice from \cite[Fig. 2]{playing}.}
\label{fig:post}
\end{figure}
In fact, the set of all relational clones (also called ``co-clones'') is   well understood.
They are listed in \cite[Table 2]{CKZ}, which gives a plain basis for each relational clone.
There is a similar table in~\cite{vollmer} (though the bases given there are not plain).
A Hasse diagram illustrating the inclusions between
the relational clones is  depicted in \cite[Figure 2]{playing}.
This diagram is reproduced here as Figure~\ref{fig:post}.
A downwards edge from one clone to another
indicates that the lower clone is a subset of the higher one.
For example, since there is a path (in this case, an edge) from $\IDo$ down to~$\IRt$ in Figure~\ref{fig:post}, we deduce that $\IRt \subset \IDo$.
We will not require bases for all relational clones,
but we have reproduced the part of \cite[Table 2]{CKZ} that we use here
as~Table~\ref{tab:post}.

\begin{table}
\begin{tabular}{||l||l||}
\hline
 & Plain Basis \\
\hline
$\IRt$ & $\{\EQ,\unaryzero,\unaryone\}$ \\
$\IDo$ & $\{\EQ,\NEQ,\unaryzero,\unaryone\}$ \\
$\ILt$ & $\{\{(x_1,\ldots,x_k)\in\{0,1\}^k\mid x_1+\ldots+ x_k=c\pmod 2\}\mid k\in\mathbb N, c\in\B\}$ \\
\hline
\end{tabular}
\medskip
\caption{The relevant portion of Table~2 of \cite{CKZ}:
Some relational clones and their plain bases.}
\label{tab:post}
\end{table}

If $R_F$ is affine then the relations in~$C$ are given by linear equations,
so $C$ is either the relational clone $\ILt$ (whose plain basis the set of all Boolean linear equations)
or $C$ is some subset of $\ILt$, in which case it is below $\ILt$ in Figure~\ref{tab:post}.

Now, $\EQ$, $\unaryzero$ and $\unaryone$ are in~$C$.
The relational clone containing  these relations (and nothing else) is $\IRt$, so
$C$ is a (not necessarily proper) superset of $\IRt$.
Thus, $C$ is (not necessarily strictly) above $\IRt$ in Figure~\ref{tab:post}.
From the figure, it is  clear that the only possibilities are that $C$ is one of
the relational clones $\ILt$, $\IDo$  and $\IRt$.

Now $\IRt \subset\IDo$ and the plain basis of $\IDo$ is
$\{\EQ,\NEQ,\unaryzero,\unaryone\}$.   Therefore
if $C=\IRt$ or $C=\IDo$, then
$R_F$ is definable by a CNF formula over
$\{\EQ, \NEQ, \unaryzero, \unaryone\}$.

Suppose that $F(\vecx)$ has arity~$n$ and, to avoid trivialities,
that $R_F$ is not the empty relation.
Suppose that $\psi(v_1,\ldots,v_n)$ is a
CNF formula over
$\{\EQ, \NEQ, \unaryzero, \unaryone\}$ implementing the
relation $R_\psi = R_F$.

Let $V = \{v_1,\ldots,v_n\}$.
Let $\psi_i$ be the projection of~$\psi$ onto variable~$v_i$.
$\psi_i$ is one of the three unary relations $\{(0)\}$, $\{(1)\}$, and $\{(0),(1)\}$.
Let
$V' = \{v_i \in V \mid \psi_i = \{(0),(1)\} \}$. ($V'$ is the set of variables that are not
pinned
in $R_F$.)
For $v_i \in V'$ and $v_j \in V'$,
let $\psi_{i,j}$ be the projection of~$\psi$ onto variables~$v_i$ and~$v_j$.
$\psi_{i,j}$ is a binary relation.
Of the 16 possible binary relations,
the only ones that can occur are $\EQ$, $\NEQ$ and $\B^2$.
The empty relation is ruled out since $R_F$ is not empty.
The four single-tuple binary relations are ruled out since $v_i$ and $v_j$ are in $V'$.
For the same reason, the other four two-tuple binary relations are ruled out.
The three-tuple binary relations are ruled out since $\psi_{i,j} \in \IDo$.

We define an equivalence relation~$\sim$ on $V'$ in which
$v_i\sim v_j$ iff $\psi_{i,j} \in \{\EQ,\NEQ\}$.
Let $V''$ contain exactly one variable from each equivalence class in~$V'$. Let $k = |V''|$.
For convenience, we will assume $V'' = \{v_1,\ldots,v_k\}$. (This can be achieved by renaming variables.)

Now, for every assignment
$\vecx:\{v_1,\ldots,v_k\}\to\B$
there is exactly one assignment
$\vecy:\{v_{k+1},\ldots,v_n\} \to \B$
such that $R_F(\vecx,\vecy) = 1$.
Let $\sigma(\vecx)$ be this assignment~$\vecy$.
Now, define the arity-$k$ function~$G$
by $G(\vecx) = F(\vecx,\sigma(\vecx))$.
Note that
\begin{equation}
\label{eq:G}
G(\vecx) = \sum_{\vecy \in\B^{n-k}} F(\vecx,\vecy),
\end{equation}
where $\vecy$ is an assignment  $\vecy:\{v_{k+1},\ldots,v_{n} \}\to\B$.
Clearly, from (\ref{eq:G}), $G \in \fclonelimeff{F,  \codomaineff_1}$.
Also, by construction, $G(\vecx)$ is a permissive function so we can apply~Lemma~\ref{lem:Topkis}. We finish with two cases.

{\it Case 1.\enspace}
Every 2-pinning of $G$ is log-modular.
Then $G$ is log-modular, by Lemma~\ref{lem:Topkis}.
This means (see, for example,~\cite[Proposition 24]{BorosHammer})
that $g = \log_2 G$ is an affine function of $x_1,\ldots,x_k$ and $\bar{x}_1,\ldots, \bar{x}_k$ so $G\in \fclone{\NEQ,\codomaineff_1}$.
For example, if $g = a_0+a_1 x_1 + a_2 x_2 + a_3\bar{x}_3$
then $G$ can be written as
\[G(x_1,x_2,x_3) = \sum_{y_3} U_0\,U_1(x_1) U_2(x_2) U_3(y_3) \NEQ(x_3,y_3),\]
where $U_0=2^{a_0}\in\codomaineff_0$ and $U_i(x) = 2^{a_i x}\in\codomaineff_1$. Since $F(\vecx,\vecy) = R_F(\vecx,\vecy) G(\vecx)$, we conclude that $F\in \fclone{\NEQ,\codomaineff_1}$.

{\it Case 2.\enspace} There is a 2-pinning $G'$ of $G$ that is not log-modular.
Since $G$ is strictly positive, so is $G'$.
Since $G \in \fclonelimeff{F,  \codomaineff_1}$,
so is $G'$.
 By Lemma~\ref{lem:binarycases},
(parts (iv) or (v)), $\IMP\in \fclonelimeff{G',\codomaineff_1}$.
By Lemma~\ref{lem:fclonelimeffTransitive}, $\IMP\in\fclonelimeff{F,\codomaineff_1}$.

Finally, we consider the case in which $C=\ILt$.
Let $\oplus_3$ be the relation
$\{(0,0,0),\allowbreak(0,1,1),\allowbreak(1,0,1),\allowbreak(1,1,0)\}$ containing all triples whose Boolean sums are~$0$.
From the plain basis of $\ILt$ (Table~\ref{tab:post}),
we see that the relation~$\oplus_3$ is
in $C$, so  $\fclone {F,\codomaineff_1}$ contains a function $F'$ with $R_{F'} = \oplus_3$.
Let $F''$ be the symmetrisation of $F'$ implemented by
$$F''(x,y,z) = F'(x,y,z)F'(x,z,y) F'(y,x,z)F'(y,z,x)F'(z,x,y)F'(z,y,x).$$
Now let $\mu_0 = F''(0,0,0)$ and $\mu_2 = F''(0,1,1)$.
Let $U$ be the unary function with $U(0)=\mu_0^{-1/3}$ and
$U(1) = \mu_0^{1/6} \mu_2^{-1/2}$.
Note that since $F\in \codomaineff$, the
appropriate roots of $\mu_0$ and $\mu_2$ are efficiently  computable, so
$U\in\codomaineff_1$.
Now
$\oplus_3(x,y,z) = U(x) U(y) U(z) F''(x,y,z)$,
so $\oplus_3 \in \fclone{F,\codomaineff_1} $.
Finally, let $U'$ be the unary function defined by $U'(0)=1$ and $U'(1)=2$
and let $G(x,z) = \sum_y \oplus_3(x,y,z) U'(y)$.
Note that $G(0,0) = G(1,1)=1$ and $G(0,1)=G(1,0)=2$.
By Lemma~\ref{lem:fcloneTransitive}, $G$ is in $\fclone{F,\codomaineff_1} $.
But by Lemma~\ref{lem:binarycases},
$\IMP \in \fclonelimeff{G,\codomaineff_1}$
so by Lemma~\ref{lem:fclonelimeffTransitive},
$\IMP \in \fclonelimeff{F,\codomaineff_1}$.

We now prove Part~2 of the theorem.
Suppose that $F$ is not lsm and
that $F\notin\fclone{\NEQ,\codomaineff_1}$
so, by Part~1 of the theorem,
we have $\IMP\in\fclonelimeff{F,\codomaineff_1}$.
Let $$H(x_1,x_2) = \sum_{y_1,y_2} \IMP(y_1,x_1)\IMP(y_1,x_2)\IMP(x_1,y_2)\IMP(x_2,y_2).$$
Note that $H(0,0)=H(1,1)=2$ and $H(0,1)=H(1,0)=1$.
Now for any integer~$k$, let
$$H_k(x_1,\ldots,x_n) = \sum_{y_1,\ldots,y_n} F(y_1,\ldots,y_n) \prod_{i=1}^n H(x_i,y_i)^k.$$
By construction, $H_k$ is strictly positive.
Also, as $k$ gets large,
$H_k(x_1,\ldots,x_n)$ gets closer and closer to  $2^{kn} F(x_1,\ldots,x_n)$.
Thus, for sufficiently large~$k$, $H_k$ is not lsm.
By Lemma~\ref{lem:fcloneTransitive},
$H\in \fclonelimeff{F,\codomaineff_1}$ so
$H_k\in \fclonelimeff{F,\codomaineff_1}$.
Applying Lemma~\ref{lem:Topkis} to $H_k$,
there is a binary function $F_1\in\fclone{F,\codomaineff_1}$
that is not lsm
so $$F_1(0,0)F_1(1,1) < F_1(0,1)F_1(1,0).$$
By Parts (iii) and (v) of Lemma~\ref{lem:binarycases}, we either have
$\NEQ\in \fclone{F,\codomaineff_1}$
or $\OR\in \fclone{F,\codomaineff_1}$. In the latter case, we are finished by Lemma~\ref{lem:OR}.
In the former case, we are also finished
since (in the notation of the proof of Lemma~\ref{lem:binarycases}) $M(\NEQ) M(\IMP) = M(\OR)$ so  $\OR\in\fclone{\IMP,\NEQ}$.
\end{proof}

\section{Complexity-theoretic consequences}\label{sec:complexity}

In order to explore the consequences of Theorem~\ref{thm:main} for the
computational complexity of approximately evaluating $\nCSP$s, we
recall the following definitions of FPRASes and AP-reductions from
\cite{DGGJ}.

For our purposes, a counting problem is a function $\varPi$ from instances~$w$
(encoded as a word over some alphabet~$\varSigma$) to a number $\varPi(w)\in\Rnonneg$.  For example,
$w$ might encode  an instance~$I$ of a counting CSP problem $\nCSP(\varGamma)$, in which
case $\varPi(w)$ would be the partition function $Z(I)$ associated with~$I$.
A \emph{randomised approximation scheme\/} for~$\varPi$
is a randomised algorithm that takes an instance~$w$
and returns an approximation $Y$ to $\varPi(w)$.
The approximation scheme has a parameter~$\varepsilon>0$
which specifies the error tolerance.  Since the algorithm is randomised,
the output~$Y$ is a random variable depending on the ``coin tosses''
made by the algorithm.  We require that, for every instance~$w$
and every $\varepsilon>0$,
\begin{equation}
\label{eq:FPRASerrorprob}
\Pr \big[e^{-\varepsilon} \varPi(w)\leq Y \leq
e^\varepsilon \varPi(w)\big]\geq \nicefrac34\, .
\end{equation}
The randomised approximation scheme is said to be a
\emph{fully polynomial randomised approximation scheme},
or \emph{FPRAS},
if it runs in time bounded by a polynomial
in~$|w|$ (the length of the word~$w$) and $\epsilon^{-1}$.
See Mitzenmacher and Upfal~\cite[Definition 10.2]{MU05}.
Note that the quantity $\nicefrac34$ in
Equation~(\ref{eq:FPRASerrorprob})
could be changed to any value in the open
interval $(\nicefrac12,1)$ without changing the set of problems
that have randomised
approximation schemes~\cite[Lemma~6.1]{JVV86}.

Suppose that $\varPi_1$ and $\varPi_2$ are functions from
$\varSigma^{\ast }$ to~$\Rnonneg$.  An ``approx\-imation-preserving reduction''
(AP-reduction)~\cite{DGGJ}
from~$\varPi_1$ to~$\varPi_2$ gives a way to turn an FPRAS for~$\varPi_2$
into an FPRAS for~$\varPi_1$.  Specifically, an {\it AP-reduction
from $\varPi_1$ to~$\varPi_2$} is a randomised algorithm~$\mathcal{A}$ for
computing~$\varPi_1$ using an oracle\footnote{The
reader who is not familiar with oracle Turing machines
can just think of this as an imaginary (unwritten)
subroutine for computing~$\varPi_2$.}
for~$\varPi_2$. The algorithm~$\mathcal{A}$ takes
as input a pair $(w,\varepsilon)\in\varSigma^*\times(0,1)$, and
satisfies the following three conditions: (i)~every oracle call made
by~$\mathcal{A}$ is of the form $(v,\delta)$, where
$v\in\varSigma^*$ is an instance of~$\varPi_2$, and $0<\delta<1$ is an
error bound satisfying $\delta^{-1}\leq\poly(|w|,
\varepsilon^{-1})$; (ii)~the algorithm~$\mathcal{A}$ meets the
specification for being a randomised approximation scheme for~$\varPi_1$
(as described above) whenever the oracle meets the specification for
being a randomised approximation scheme for~$\varPi_2$; and (iii)~the
run-time of~$\mathcal{A}$ is polynomial in $|w|$ and
$\varepsilon^{-1}$.  Note that
the class of functions computable by an FPRAS is closed
under AP-reducibility.
Informally, AP-reducibility is the most liberal notion
of reduction meeting this requirement.
If an AP-reduction from $\varPi_1$ to~$\varPi_2$
exists we write $\varPi_1\APred\varPi_2$.
If $\varPi_1\APred\varPi_2$ and $\varPi_2\APred\varPi_1$ then we say that
{\it $\varPi_1$ and $\varPi_2$ are AP-interreducible}, and write $\varPi_1\APeq\varPi_2$.

A word of warning about terminology.  Subsequent to~\cite{DGGJ} the notation~$\APred$ has been used
to denote a different type of approximation-preserving reduction which applies to optimisation
problems. We will not study optimisation problems in this paper, so hopefully this will not
cause confusion.

The complexity of approximating Boolean \#CSPs in the unweighted case,
where the functions in $\varGamma$ have codomain $\B$,
was studied earlier~\cite{trichotomy} by some of the authors  of this paper.
Two counting problems played a special role there, and in
previous work in the complexity of approximate counting~\cite{DGGJ}.
They also play a key role here.

\prob{$\SAT$}{A Boolean formula $\varphi$
in conjunctive normal form.}{The number of satisfying
assignments of~$\varphi$.}

\prob{$\BIS$.}{A bipartite graph $B$.}{The number of independent sets
in~$B$.}

An FPRAS for $\SAT$ would, in particular, have to decide with
high probability between a formula having some satisfying assignments
or having none.  Thus $\SAT$ cannot have an FPRAS unless
$\NP=\RP$.\footnote{The supposed FPRAS would provide
a polynomial-time decision procedure for satisfiability with two-sided
error;  however, there is a standard trick for converting two-sided
error to the one-sided error demanded by the definition of~$\RP$ \cite[Thm 10.5.9]{Weg}.}
The same is true of any problem to which $\SAT$
is AP-reducible.  As far as we are aware,
the complexity of approximating $\BIS$ does not
relate to any of the standard complexity theoretic assumptions,
such as $\NP\not=\RP$.  Nevertheless, there is increasing
empirical evidence that no FPRAS for $\BIS$ exists, and we adopt this
as a working hypothesis.  Of course, this hypothesis implies that
no $\BIS$-hard problem (problem to which $\BIS$ is AP-reducible)
admits an FPRAS\null.

Finally, a precise statement of the computational task we are
interested in.
A (weighted) \#CSP problem is parameterised by a finite subset~$\mathcal{F}$
of~$\codomaineff$ and defined as follows.

\prob{$\nCSP(\calF)$} {A \ppformula{} $\psi$ consisting of a product of $m$~atomic $\calF$-formulas
over $n$~free variables~$\vecx$.  (Thus, $\psi$ has no bound variables.)} {The value
$\sum_{\vecx\in\B^n} F_\psi(\vecx)$ where $F_\psi$ is
the function defined by that formula.}

Officially, the input size $|w|$ is the length of the encoding of the instance.
However, we shall take the size of a $\nCSP(\calF)$ instance to be
$n+m$, where $n$ is the number of (free) variables and $m$ is the number of
constraints (atomic formulas).  This is acceptable, as we are only concerned to
measure the input size within a polynomial factor; moreover, we have restricted~$\calF$
to be finite, thereby avoiding the issue of how to encode the constraint functions~$\calF$.
We typically denote an instance of $\nCSP(\calF)$ by~$I$
and the output by $Z(I)$;  by analogy with systems in statistical physics
we refer to $Z(I)$ as the partition function.

Aside from simplifying the representation of problem instances,
there is another, more important reason for decreeing that $\calF$ is finite, namely, that it allows us to prove the following basic lemma
relating functional clones and computational complexity.
It is, of course, based on a similar result for classical decision CSPs.

\begin{lemma}\label{obsAPclones}
Suppose that $\calF$ is a finite subset of $\codomaineff$.
If  $F\in \fclonelimeff\calF$ then
$$\nCSP(F, \calF) \APred \nCSP(\calF).$$
\end{lemma}
\begin{proof}
Let $k$ be the arity of~$F$. Let $\calM$ be a TM which, on input $\epsilon'>0$,
computes a $k$-ary \ppformula~$\psi$ over $\calF \cup \EQ$ such that
  $\| F_\psi-F\|_\infty<\epsilon'$.
We can assume without loss of generality that no
function in  $\{F\} \cup \calF$ is identically zero
(otherwise every \#CSP instance using this
function has partition function~$0$).
Let $\mumax$ be the maximum value in the range of~$F$
and let $\mumin$ be the minimum of~$1$ and
the minimum non-zero value in the range of~$F$.
Similarly, let~$S$ be the set of non-zero values in the range of functions in
$\calF\cup \{\EQ\}$. Let $\numax$ be the maximum value in~$S$ and let
$\numin$ be the minimum of~$1$ and the minimum value in~$S$.

Consider an input $(I,\epsilon)$
where $I$ is an instance of $\nCSP(F, \calF)$ and $\epsilon$ is an accuracy parameter.
Suppose that $I$ has
$n$ variables, $m$ $F$-constraints,
and $m'$ other constraints.
We can assume without loss of generality that $m>0$ (otherwise, $I$ is an instance of
$\nCSP(\calF)$).

The key idea of the proof is to construct an instance $I'$ of $\nCSP(\calF)$ by replacing each
$F$-constraint in~$I$ with the set of constraints and extra (bound) variables in
the formula~$\psi$
that is output by~$\calM$ with input~$\epsilon'$.
We determine how small to make $\epsilon'$ in terms of the following quantities.
Let
\begin{align*}
A &= \frac{4 m}{\mumin} 2^n \mumax^m \numax^{m'}\\
B &= 2^n {(\mumax+1)}^{m-1} \numax^{m'}\\
C &= \mumin^m \numin^{m'}.
\end{align*}
Let $\epsilon' = \tfrac \epsilon4 \tfrac{C}{A+B}$.
The time needed to construct~$\psi$ for a given $\epsilon'>0$ is
at most~$\poly(\log {(\epsilon'}^{-1}))$, which is at most a polynomial in~$n$, $m$, $m'$ and
$\epsilon^{-1}$, as required by the definition of AP-reduction.
We shall see that $(I',\epsilon/2)$ is the sought-for
instance/tolerance pair required by our reduction.

Let $I_\psi$ be the instance formed from~$I$ by replacing every
$F$-constraint with an $F_{\psi}$-constraint.  Note that $Z(I_\psi)=Z(I')$,
since $I'$, an instance of $\nCSP(\calF)$, is an implementation of~$I_\psi$.
We want to show that if an oracle produces a sufficiently accurate
approximation to $Z(I')$ (and hence to $Z(I_\psi)$) then we can deduce a
sufficiently accurate approximation to~$Z(I)$. Observe that the definition of
FPRAS allows no margin of error when $Z(I)=0$, and our reduction must give
the correct result, namely~$0$, in this case.  Therefore we need to treat
separately the cases $Z(I)=0$ and $Z(I)>0$. We will show that
\begin{equation}\label{wish1}
Z(I)=0 \quad\text{implies}\quad Z(I_\psi)<C/3,
\end{equation}
and
\begin{equation}\label{wish2}
Z(I)>0\quad\text{implies}\quad Z(I_\psi)>2C/3;
\end{equation}
moreover, in the latter case,
\begin{equation}
\label{wish}
e^{-\epsilon/2} Z(I) \leq Z(I_\psi) \leq e^{\epsilon/2} Z(I).
\end{equation}

These estimates are enough to ensure correctness of the reduction.
For a call to an oracle for $\nCSP(\calF)$ with instance~$I'$
and accuracy parameter~$\epsilon/2$ would return a result in the
range $[e^{-\epsilon/2}Z(I_\psi),e^{\epsilon/2} Z(I_\psi)]$
with high probability.  Observe that this estimate is sufficient to distinguish
between cases (\ref{wish1}) and (\ref{wish2}).  In the former case,
we are able to return the exact result, namely~$0$.  In the latter
case, we return the result given by the oracle, which by~(\ref{wish})
satisfies the conditions for an FPRAS\null.

To establish~(\ref{wish1}--\ref{wish}), let $Y'$ be the set of
assignments to the variables of instance~$I$ which make a non-zero
contribution to~$Z(I)$ and let $Y''$ be the remaining assignments to the
variables of instance~$I$. Let $Z'(I_\psi)$ be the contribution to
$Z(I_\psi)$ due to assignments in~$Y'$ and $Z''(I_\psi)$ be the
contribution to $Z(I_\psi)$ due to assignments in~$Y''$ (so $Z(I_\psi) =
Z'(I_\psi) + Z''(I_\psi)$). We can similarly write $Z(I) = Z'(I) +
Z''(I)$, though of course $Z''(I)=0$.

First, note that if
$|F_\psi(\vecx) -  F(\vecx)|\leq \epsilon'$ and $F(\vecx)>0$ then
$$|F_\psi(\vecx)/F(\vecx)-1| \leq \epsilon'/F(\vecx) \leq \epsilon'/\mumin,$$
so
$$e^{-2\epsilon'/\mumin} \leq \frac{F_\psi(\vecx)}{F(\vecx)} \leq e^{2 \epsilon'/\mumin}.$$
We conclude that
$$e^{-2\epsilon' m/\mumin} Z'(I) \leq Z'(I_\psi)  \leq e^{2 \epsilon' m/\mumin} Z'(I),$$
so
\begin{equation*}
|Z'(I)-Z'(I_\psi)| \leq \frac{4 \epsilon' m}{\mumin} Z(I) \leq \epsilon' A.
\end{equation*}
Furthermore,
\begin{equation*}
|Z''(I)-Z''(I_\psi)| = Z''(I_\psi) \leq \epsilon' B.
\end{equation*}
Here we use $\|F_\psi-F\|_\infty<\epsilon'<1$;  the ``$\null+1$'' in the definition
of~$B$ absorbs the discrepancy between $F_\psi$ and~$F$.
Combining these two inequalities yields
\begin{equation}
\label{11aug}
|Z(I)-Z(I_\psi)| \leq \epsilon' (A+B)\leq \frac{\epsilon C}4.
\end{equation}

Now, $Z(I)>0$ implies $Z(I)\geq C$, and hence (\ref{wish1}) and (\ref{wish2})
follow directly from~(\ref{11aug}).
If $Z(I)>0$ we further have
$$\left|\frac{Z(I_\psi)}{Z(I)}-1\right| \leq \frac{\epsilon'(A+B)}{Z(I)} \leq \frac{\epsilon'(A+B)}{C} \leq \epsilon/3.$$
This establishes~(\ref{wish}) and completes the verification of the reduction.
\end{proof}

\begin{theorem}
\label{corthm:main}
Suppose $\calF$ is a finite subset of
$\codomaineff$.
\begin{itemize}
\item If $\calF \subseteq \fclone{\NEQ,\codomaineff_1}$
then, for any finite subset~$S$ of $\codomaineff_1$,
there is an FPRAS for $\nCSP(\calF,S)$.
\item Otherwise,
\begin{itemize}
\item[$\circ$]There is a finite subset~$S$ of $\codomaineff_1$ such that
$\BIS \APred \nCSP(\calF, S)$.
\item[$\circ$] If there is a function $F\in \calF$ such that $F\notin \LSM$
then  there is a finite subset~$S$ of $\codomaineff_1$ such that
$\SAT \APeq \nCSP(\calF, S)$.
\end{itemize}
\end{itemize}
\end{theorem}

\begin{proof}
First, suppose that $\calF \subseteq  \fclone{\NEQ,\codomaineff_1}$.
Let $S$ be a   finite subset  of~$\codomaineff_1$.
Given an $m$-constraint input $I$ of $\nCSP(\calF, S)$
and an accuracy parameter~$\epsilon$,
we first approximate  each arity-$k$ function $F\in \calF$
used in~$I$
with
a function $\hatF:\B^k \rightarrow \mathbb{Q}^{\geq 0}$
such that  $R_{\hatF} = R_F$, and for every $\vecx$ for which $F(\vecx)>0$,
\begin{equation}
\label{eq:approx}
e^{-\epsilon/m} \leq
\frac{\hatF(\vecx)}{F(\vecx)} \leq e^{\epsilon/m}.
\end{equation}
Let $\widehat\calF = \{ \hatF \mid F \in \calF\}$
and let $\widehat I$ be the instance of $\nCSP(\widehat\calF, S)$ formed from~$I$
by replacing each $F$-constraint with $\hatF$.
\cite[Theorem 4]{wbool} gives a polynomial-time algorithm for computing
the partition function $Z(\widehat I)$, which satisfies
\begin{equation}
\label{eq:another}
e^{-\epsilon} Z(I) \leq Z(\widehat I) \leq e^{\epsilon} Z(I).
\end{equation}

Second, suppose that $\calF \not\subseteq  \fclone{\NEQ,\codomaineff_1}$.
By Theorem~\ref{thm:main}, $\IMP\in\fclonelimeff{\calF,\codomaineff_1}$.
By Observation~\ref{obs:26July}, there is a finite subset~$S$ of~$  \codomaineff_1$ such
that $\IMP\in\fclonelimeff{\calF, S}$.
Thus, $\nCSP(\IMP) \APred \nCSP(\calF,S)$, by Lemma~\ref{obsAPclones}.
However, $\BIS \APeq \nCSP(\IMP)$ by \cite[Theorem 3]{trichotomy}.

Finally, suppose that  there is a function $F\in \calF$ such that $F\notin \LSM$. By Theorem~\ref{thm:main},
$\fclonelimeff{F,\codomaineff_1}=\codomaineff$ so
$\OR \in \fclonelimeff{F,\codomaineff_1}$.
By Observation~\ref{obs:26July},
there is a finite subset~$S$ of~$  \codomaineff_1$ such
that $\OR\in\fclonelimeff{F, S}$,
so by Lemma~\ref{obsAPclones},
$\nCSP(\OR) \APred \nCSP(F,S)$. However, by \cite[Lemma 7]{trichotomy}
$\SAT \APred \nCSP(\OR)$.
To see that $\nCSP( \calF, S) \APred \SAT$,
let $I$ be an $m$-constraint instance of $\nCSP( \calF, S)$.
For each function $G\in \calF\cup S$,
define $\widehat{G}$ as in~(\ref{eq:approx}).
Let $\widehat I$ be
the instance of
$\nCSP( \{\widehat{G} \mid G \in \calF \cup S\} )$
formed from~$I$ by replacing each $G$-constraint with a $\widehat{G}$-constraint.
Equation~(\ref{eq:another}) holds, as above.
Furthermore, from \cite[Section 1.3]{wbool} the problem
of evaluating  $\nCSP(
\{\widehat{G} \mid G \in \calF \cup S\})$
(even with $\widehat{G}$ as part of the input)
is in $\numP_\Q$, the complexity class comprising functions which are a
function in $\numP$ divided by a function in FP,
so can be AP-reduced to $\SAT$.
\end{proof}

\begin{example}
Let $F\in\codomaineff_2$ be the function defined by $F(0,0)=F(1,1)=\lambda$
and $F(0,1)=F(1,0)=1$, where $\lambda>1$.  Then, from Theorem~\ref{corthm:main},
$\nCSP(F,S)$ is $\BIS$-hard, for some set $S$ of unary weights.
(In fact, this counting CSP is also $\BIS$-easy.) Note that $\nCSP(F, S)$
is nothing other than the ferromagnetic Ising model with an applied field.
So we recover, with no effort, the main result of Goldberg and Jerrum's
investigation of this model~\cite{GJ07}.
\end{example}

\begin{example}
If $F$ is as before, but $\lambda\in(0,1)$, then $F\notin\LSM$ and
Theorem~\ref{corthm:main} tells us that $\nCSP(F, S)$ is $\SAT$-hard, for
some set $S$ of unary weights. This is a restatement of the known fact that
the partition function of the antiferromagnetic Ising model is hard to
approximate~\cite{JS93}. Actually, this is true without weights, but that is not directly implied by Theorem~\ref{corthm:main}.
\end{example}

\section{The classes $\fclone{\LSM_k}$}\label{sec:cloneLSMk}
In this section, we will be concerned with expressibility, and less so with efficient computability. Thus, we will use function classes  without attempting to distinguish the efficiently computable functions in the class from the remainder.

Define $\LSM_k=\LSM\cap\codomainbool_k$. It follows from the proof of Lemma~\ref{lem:binarycases}
that $\fclone{\LSM_2}=\fclone{\IMP,\codomainbool_1}$. Since $\LSM$ is central to Theorem~\ref{thm:main}, we need to consider whether
$\LSM=\fclonelim{\LSM_2}$ and, if not, what internal structure it may possess. To this end, we  consider here the functional clones $\fclonelim{\LSM_k}$ for $k>2$.
As mentioned in the abstract, we will make use of two transforms: the M\"obius transform to show
$\fclonelim{\LSM_3}=\fclonelim{\LSM_2}$ (in fact we
prove the stronger statement $\fclone{\LSM_3}=\fclone{\LSM_2}$),
and the Fourier transform to show that $\fclonelim{\LSM_4}\neq \fclonelim{\LSM_2}$. 
See \cite{Billionnet,ZCJ09,ZJ10} for corresponding results in the context of optimisation.

For all $\vecx,\vecy\in\B^n$, we will write $\vecx\leq \vecy$
to mean that for all $1\leq i\leq n$ we have $x_i\leq y_i$.
For all $f\in\codomainreal_n$ define the \emph{M\"obius transform}
$\widetilde{f}\in\codomainreal_n$ by
\begin{align}
\widetilde{f}(\vecy)\,&=\,\sum_{\vecw\leq\vecy} (-1)^{|\vecy-\vecw|}f(\vecw)&& (\vecy\in\B^n),\label{eq:coefficients3}
\intertext{and note that the M\"obius transform is invertible:}
f(\vecx)\,&=\,\sum_{\vecy\leq\vecx}\widetilde{f}(\vecy) && (\vecx\in\B^n).\label{eq:coefficients2}
\end{align}

See also~\cite{GrMaRo00} for further
information. We will not require anything other than
\eqref{eq:coefficients3} and \eqref{eq:coefficients2}.
We next show that certain simple functions are in $\fclone{\LSM_2}$.
\begin{lemma}\label{lem:chi}
Let $\vecy\in\{0,1\}^n$. Let $t\in\mathbb{R}$, 
with $t\geq 0$ if $|\vecy|>1$.
Let $F\in\codomainbool_n$ be the unique function
satisfying $\widetilde{\log F}(\vecy)=t$ and $\widetilde{\log
  F}(\vecx)=0$ for $\vecx\neq\vecy$.  
(Explicitly, 
$F(\vecx)=e^t$ for
$\vecy\leq\vecx$ and $F(\vecx)=1$ otherwise.)  Then
$F\in\fclone{\LSM_2}$.
\end{lemma}
\begin{proof}
If $t\geq 0$ define $U\in\LSM_1$ by $U(0)=1$ and $U(1)=e^t-1$. We will
argue that for all $\vecx$ we have
\[ F(\vecx)=\sum_z U(z) \prod_{i:y_i=1} \IMP(z,x_i). \]
Indeed if $\vecy\leq\vecx$ we get $F(\vecx)=U(0)+U(1)=e^t$, and otherwise
$F(\vecx)=U(0)=1$.

Now we consider the case $|\vecy|\leq 1$.  If $\vecy=\constzero$ let
$i=1$, and otherwise let $i$ be the unique index with $y_i=1$.  Then
$F(\vecx)=U(x_i)$ where $U(0)=F(\constzero)$ and $U(1)=F(\constone)$.
Hence $F\in\fclone{\LSM_1}\subset\fclone{\LSM_2}$.
\end{proof}

If $\vecx\in\B^n$, let $\bar{\vecx}=\constone-\vecx$ and, if $F\in\codomainbool_n$, let $\bar{F}(\vecx)=F(\bar{\vecx})$. We will use the following simple fact.
\begin{lemma}\label{lem:lsmk10}
$F\in\LSM_k$ iff $\bar{F}\in\LSM_k$.
\end{lemma}
\begin{proof}
   By symmetry, it suffices to show that, for any $\vecx,\vecy\in\B^k$,
   \begin{align*}
    \bar{F}(\vecx)\bar{F}(\vecy)\,&=\,F(\bar{\vecx})F(\bar{\vecy})\,\leq\, F(\bar{\vecx}\wedge\bar{\vecy})F(\bar{\vecx}\vee\bar{\vecy})\\
    &=\,F(\,\Bar{\!\vecxuy\!}\,)F(\,\Bar{\!\vecxny\!}\,)
    \,=\,\bar{F}(\vecxuy)\bar{F}(\vecxny).\qedhere
   \end{align*}
\end{proof}

Now we show that it is only necessary to consider permissive
functions.  The construction used in the lemma is adapted from one
given, in a more general setting, by Topkis~\cite{Topkis}.

\begin{lemma}\label{lem:lifting}
For every $F\in\LSM_k$ there exists $G\in\LSMp_k$ such that $F=R_FG$.
Furthermore $R_F\in\fclone{\LSM_2}$, so $\fclone{\LSM_k}=\fclone{\LSMp_k,\LSM_2}$.
\end{lemma}
\begin{proof}
First, assume $F(\constzero)\neq 0$.  Let $\mu=\min_{\vecx}\{F(\vecx)\mid
F(\vecx)\neq 0\}/\max_{\vecx} F(\vecx)$.  Set
\[ G(\vecx)=\max \{F(\vecy)\mu^{|\vecx|-|\vecy|}\mid \vecy\leq \vecx\}.\]
Then $G$ is strictly positive and $G(\vecx)=F(\vecx)$ wherever
$F(\vecx)\neq 0$. It remains to show that $G\in\LSM$. For all
$\vecx,\vecx'$ there exist $\vecy\leq\vecx$ and $\vecy'\leq\vecx'$
such that
\begin{align*}
G(\vecx)G(\vecx')&=F(\vecy)F(\vecy')\mu^{|\vecx|-|\vecy|+|\vecx'|-|\vecy'|}\\
&\leq F(\vecy\wedge\vecy')F(\vecy\vee\vecy')\mu^{|\vecx|-|\vecy|+|\vecx'|-|\vecy'|}\\
&= F(\vecy\wedge\vecy')F(\vecy\vee\vecy')\mu^{|\vecx\wedge\vecx'|-|\vecy\wedge\vecy'|+|\vecx\vee\vecx'|-|\vecy\vee\vecy'|}\\
&\leq G(\vecx\wedge\vecx')G(\vecx\vee\vecx').
\end{align*}

Now we deal with the case $F(\constzero)=0$. Let $F'(\vecx)=F(\vecx)$
for all $\vecx\neq \constzero$ and let $F'(\constzero)=1$. Then
$F'\in\LSM$ and $F'(\constzero)\neq 0$, and we have shown that there
exists $G\in\LSMp$ such that $F'=R_{F'}G$. But then $F=R_FG$.

By \cite[Corollary 18]{trichotomy}, $R_F$ is a conjunction of
implications and constants, and hence
$R_F\in\fclone{\unary_0,\unary_1,\IMP}\subset\fclone{\LSM_2}$. Thus
$\fclone{\LSM_k}\subseteq\fclone{\LSMp_k,\LSM_2}$. The reverse inclusion is trivial.
\end{proof}

\begin{lemma}\label{lem:IMParitythree}
$\fclone{\LSM_3}=\fclone{\LSM_2}$.
\end{lemma}
\begin{proof}
From Lemma~\ref{lem:lifting}, we need only prove $\LSMp_3\subseteq\fclone{\LSM_2}$. Thus
let $F\in\LSMp_3$, $f=\log F$, and 
first
assume $\widetilde{f}(1,1,1)\geq 0$.  We will show that
$F\in\fclone{\LSM_2}$.  Note that, by log-supermodularity of $F$,
    \[ \widetilde{f}(1,1,0) = f(0,0,0)-f(1,0,0) - f(0,1,0)+f(1,1,0)\geq 0.\]
and similarly $\widetilde{f}(1,0,1),\widetilde{f}(0,1,1)\geq 0$.
Hence $\widetilde{f}(\vecy)\geq 0$ for all $|\vecy|>1$.  For all
$\vecy\in\B^3$ let $F_{\vecy}$ be the unique function satisfying
$\widetilde{\log F_{\vecy}}(\vecy)=\widetilde{f}(\vecy)$ and
$\widetilde{\log F_{\vecy}}(\vecz)=0$ for $\vecz\neq\vecy$. Then
$\widetilde{\log F}=\sum_{\vecy} \widetilde{\log F_{\vecy}}$, which implies
$\log F=\sum_{\vecy} \log F_{\vecy}$, which implies $F=\prod_{\vecy}
F_{\vecy}$. By Lemma~\ref{lem:chi} we have
$F_{\vecy}\in\fclone{\LSM_2}$ for all $\vecy$, and therefore
$F\in\fclone{\LSM_2}$.
If $\widetilde{f}(1,1,1)<0$, let $H=\overline{F}$.  Note that $H\in\LSMp_3$ and $\widetilde{\log H}(1,1,1)>
0$, so by the previous paragraph, $H\in\fclone{\LSM_2}$. By
Lemma~\ref{lem:lsmk10} this implies $F\in\fclone{\LSM_2}$.
\end{proof}
In view of Lemma~\ref{lem:IMParitythree}, it might be conjectured that $\LSM=\fclonelim{\LSM_2}$. In fact,
this is not the case, as we will now show.
First we consider the class $\calP$ of functions $F\in\codomainbool$
for which the \emph{Fourier transform} $\widehat{F}$ has
nonnegative coefficients, where
\begin{align}\widehat{F}(\vecy)\,&=\,\frac{1}{2^n}\sum_{\vecw\in\B^n} (-1)^{|\vecw\wedge\vecy|}F(\vecw)&& (\vecy\in\B^n)\label{eq:coefficients5}.\end{align}
Thus $F\in\calP$ if and only if $\widehat{F}\in\codomainbool$.
See \cite{deWolf} for further information. We will use \eqref{eq:coefficients5} and the \emph{convolution theorem}: for
all $F,G\in\codomainbool_n$ we have
\begin{align}\widehat{FG}(\vecx)\,&=\,\sum_{\vecy\in\B^n} \widehat{F}(\vecy)\widehat{G}(\vecx\oplus\vecy)&&(\vecx,\vecy\in\B^n).\label{eq:convolution}
\end{align}
where $\oplus$ denotes componentwise addition modulo 2. See for
example \cite[Section 2.3]{deWolf} for a proof of the dual statement.

To show that $\calP$ is closed under \ppformula\ evaluation, it is
useful to restrict to atomic formulas where variables are not repeated within a scope.

\begin{lemma}\label{lem:no_repeated_variables}
Let $\calF\subseteq\codomainbool$.  For all \ppformulas\ $\psi$ over
$\calF$ there is another \ppformula\ $\psi'$ over $\calF\cup\{\EQ\}$ such that
$F_\psi=F_{\psi'}$ and no atomic formula of $\psi'$ contains a repeated variable.
\end{lemma}
\begin{proof}
Given $\psi$ obtain $\psi'$ as follows.  For each variable $v_i$
that is used $d_i\geq 2$ times in total in $\psi$,
replace the uses of $v_i$ by new distinct variables
$v_i^1,\cdots,v_i^{d_i}$, multiply by atomic formulas
$\EQ(v_i,v_i^j)$ for $1\leq j\leq d_i$, then sum over these new
variables $v_i^j$.
\end{proof}

\begin{lemma}\label{lem:lsmk20}
$\calP$ is closed under addition, summation, products and limits. Moreover, $\calP$ is a \omegadefinable\ functional clone.
\end{lemma}
\begin{proof}

If $F,G\in\calP$, then $\widehat{F+G}=\widehat F+\widehat G$ is clearly non-negative,
and $\widehat{FG}$ is nonnegative by the convolution theorem \eqref{eq:convolution}. For summation, as in Lemma~\ref{lem:closedUnderProjection}, we consider summing over the last variable. So, let $H(\vecx)=\sum_t F(\vecx,t)$. Then it follows easily from \eqref{eq:coefficients5} that $\widehat{H}(\vecy) = 2\widehat{F}(\vecy,0)\geq0$ for all $\vecy$.
For limits note that if $F_n\to F$ then $\widehat{F_n}\to\widehat{F}$,
and a limit
of non-negative functions is non-negative.

Let $\psi$ be a \ppformula\ over $\calP\cup\{\EQ\}$. We will argue that
that $F_\psi\in\calP$.
By Lemma \ref{lem:no_repeated_variables} there is a
\ppformula\ $\psi'$ over $\calP\cup\{\EQ\}$ such that $F_\psi=F_{\psi'}$ and such
that no atomic formula of $\psi$ contains a repeated variable.  The
functions $F_\phi$ defined by atomic formulas
$\phi=G(v_{i_1},\cdots,v_{i_k})$ of $\psi'$ are therefore
``expansions'': permutations of the function $G'\in\codomainbool_n$,
$n\geq k$, defined by
\begin{align}G'(\vecx,\vecx')\,&=\,G(\vecx) && (\vecx\in\{0,1\}^k\text{ and }\vecx'\in\{0,1\}^{n-k}).\label{eq:expansion}\end{align}
It therefore suffices to check that $\calP$ is closed under
expansions.
Let $G'$ be the expansion defined by
\eqref{eq:expansion}.
Then, for all $\vecy\in\{0,1\}^k$ and $\vecy'\in\{0,1\}^{n-k}$, we have $\widehat{G'}(\vecy,\vecy')=\widehat{G}(\vecy)$ if $\vecy'=0^k$,
and $\widehat{G'}(\vecy,\vecy')=0$ otherwise, and hence $G'\in\calP$.
Note that $\widehat{\EQ}=\frac12\EQ$, so $\EQ\in\calP$.
Thus $\calP$ is a \ppdefinable\ functional clone, but it is also closed under
limits.
\end{proof}
For $F\in\codomainbool$, let $F^\star$ denote $F\bar{F}$. Now   let $\calC$ be the class of functions $F\in\codomainbool$ such that $G^\star\in\calP$ for every pinning $G(\vecx)=F(\vecx;\vecc)$. Note, in particular, that if $U\in\codomainbool_1$, $U^\star(z)=U(0)U(1)$, a nonnegative constant. Therefore we have $\codomainbool_1\subseteq\calC$ and, to establish that $F\in\calC$, we need only check pinnings of $F$ of arity at least $2$.
\begin{lemma}\label{lem:lsmk90}
$\calC$ is a \omegadefinable\ functional clone.
\end{lemma}
\begin{proof}
As in Lemma \ref{lem:lsmk20} we will check that $\calC$ is closed under
``expansions'', products, summations, and limits.
But a pinning of an expansion (or product, summation, or limit) of functions in $\calC$ is
an expansion (or product, summation, or limit) of pinnings of functions in $\calC$,
which are necessarily in $\calC$ because $\calC$ is closed under pinnings.
So it suffices to check the $\calC$ condition for trivial pinnings,
for example to check closure under products it suffices to show
that $F,G\in\calC$ implies $(FG)^\star\in\calP$.

Let $G\in\calC$ have arity $k$, let $n\geq k$, let $G'$ be the function
defined by \eqref{eq:expansion}.
Note that $G'\Bar{G'}$ is an expansion of $G\Bar G$, so $G'\Bar{G'}\in\calP$ and $G'\in\calC$.
We have $\EQ\in\calC$, since $\EQ\hspace{1pt}\Bar{\EQ}=\EQ\in\calP$.
Closure under product follows from Lemma~\ref{lem:lsmk20} and
the observation that $(FG)^\star= F^\star G^\star$.
For summation, again as in Lemma~\ref{lem:closedUnderProjection}, we consider summing over the last variable. Then, if $H(\vecx)=\sum_tF(\vecx,t)$, where $F$ has arity $k+1$, then
\begin{equation*}
    H^\star(\vecx) = \textstyle\sum_tF(\vecx,t)\sum_t\bar{F}(\vecx,t)
    =\textstyle (F_0)^\star(\vecx) + (F_1)^\star(\vecx) + \sum_{t} F^\star(\vecx,t).
\end{equation*}
where $F_0$ and $F_1$ are the pinnings $F_i(\vecx)=F(\vecx;i)$.
We have $(F_0)^\star,(F_1)^\star\in\calP$ by the pinning assumption, and the arity $k$
function $\sum_{t}F^\star(\vecx,t)$ is in $\calP$ by~Lemma~\ref{lem:lsmk20}.
Thus 
$H^\star$ 
is the sum of three functions in $\calP$, and so, using
Lemma~\ref{lem:lsmk20} again, $H^\star\in\calP$. Finally note that $\calC$ is closed
under limits: if $F_n\to F$ as $n\to\infty$ then $F_n^\star\to F^\star$, but $\calP$ is closed under limits.
\end{proof}
\begin{lemma}\label{lem:lsmk30}
$\fclonelim{\LSM_2}\subseteq\calC$.
\end{lemma}
\begin{proof}
Let $F\in\LSM_2$. Note that $\widehat{F^\star}(0,0)=(F(0,0)F(1,1)+F(0,1)F(1,0))/2$, and
$\widehat{F^\star}(0,1)=\widehat{F^\star}(1,0)=0$, and
$\widehat{F^\star}(1,1)=(F(0,0)F(1,1)-F(0,1)F(1,0))/2\geq 0$.  So
$F^\star\in\calP$, and hence $F\in\calC$. Thus $\LSM_2\subseteq\calC$
and, since $\calC$ is a \omegadefinable\ functional clone,
$\fclonelim{\LSM_2}\subseteq\calC$.
\end{proof}
\begin{lemma}\label{lem:lsmk100}
$\fclonelim{\LSM_2}\subset\fclonelim{\LSM_4}$.
\end{lemma}
\begin{proof}
Since $\LSM_2\subseteq\calC$ by Lemma~\ref{lem:lsmk30}, we need only exhibit a function $F\in\LSM_4$ which is not in $\calC$. Define $F:\{0,1\}^4\to\Rpos$ by
\[F(x_1, x_2, x_3, x_4) =\left\{\begin{array}{ll}
                         4, & \hbox{if}\ x_1 + x_2 + x_3 + x_4 = 4; \\
                         2, & \hbox{if}\ x_1 + x_2 + x_3 + x_4 = 3;\\
                         1, & \hbox{otherwise.}
                       \end{array}
\right.\]
To show $F \in \LSM_4$, by  the symmetry of $F$ and Lemma~\ref{lem:Topkis}, it suffices to
show that the three $2$-pinnings $F(0,0,x_3,x_4)$, $F(0,1,x_3,x_4)$ and $F(1,1,x_3,x_4)$ are lsm.
This is equivalent to the inequalities $1 \times 1 \geq 1 \times 1$, $2 \times 1 \geq 1 \times 1$,
and $4 \times 1 \geq 2 \times 2$ respectively, which clearly hold.

To show that $F\notin\calC$, we need only use \eqref{eq:coefficients5} to calculate
\[ \widehat{F^\star}(1, 1, 1, 1)\, =\, \frac{4\times 1 - 4\times 2 + 6\times 1 - 4\times 2 + 4\times 1}{2^4}\, = \,-\frac{1}{8}\, <\, 0.\]

Indeed $F\in\fclonelim{\LSM_4}$ but $F\notin\calC$.
By Lemma \ref{lem:lsmk30}, $\fclonelim{\LSM_2}\subseteq\fclonelim{\LSM_4}\cap\calC\subset\fclonelim{\LSM_4}$.
\end{proof}
Unfortunately,  this approach does not seem to extend to showing $\fclonelim{\LSM_4}\subset\LSM$ or even $\fclone{\LSM_4}\subset\LSM$. Neither can we extend the result of Lemma~\ref{lem:IMParitythree} to show that $\fclonelim{\LSM_4} = \fclonelim{\LSM_5}$. However, we will venture the following, which is true for $k=1$.
\begin{conjecture}\label{lsmkconjecture}
For all $k\geq1$, $\fclonelim{\LSM_{2k}} = \fclonelim{\LSM_{2k+1}} \subset\fclonelim{\LSM_{2k+2}}$.
\end{conjecture}
A consequence of a proof of Conjecture~\ref{lsmkconjecture} would be that $\LSM\neq\fclonelim\calF$ for any finite set of functions~$\calF$.

\section{Restricted unary weights}\label{sec:restrictedWeights}
In this section and the next, we depart from the conservative case, and consider allowing only restricted classes of unary weights.

We have already noted that restricting to nonnegative (as opposed
to arbitrary real) unary weights produces a richer lattice of functional clones, and an apparently richer complexity landscape. Thus,
by further restricting the unary functions available, we might expect to further refine the lattice of functional clones.

In this section, we will consider the unary functions that favour 1 over~0, or vice versa. With a view to studying this setting, let
$\codomaindown_1$ (respectively $\codomainup_1$) be the class of unary functions $U$
from~$\codomainnneg_1$ such that $U(1)\le U(0)$ (respectively, $U(0)\le U(1)$).
By $\codomaindowneff_1$ and $\codomainup_1$ we denote the efficient versions of these sets. Also denote by $\EQp$ the \emph{permissive equality} function defined by $\EQp(x,y)=2$ if $x=y$, and $\EQp(x,y)=1$ otherwise.

The first hint that partitioning $\codomainbool_1$ into $\codomainup_1$ and
$\codomaindown_1$ may yield new phenomena comes from the observation that for any finite subset $S$ of
$\codomaindowneff_1$, the problem $\nCSP(\EQp,S)$ reduces to the
ferromagnetic Ising model
with a consistent external field, for which Jerrum and Sinclair have given an
FPRAS~\cite{JS93}. (A~consistent field is one that favours one of the two
spins consistently over every site.) The point here is that
$\nCSP(\EQp,S)$ is tractable, even though
$\EQp\notin\fclonelim{\NEQ,\codomainbool_1}$. In contrast,
by Theorem~\ref{corthm:main} or from
the arguments in ~\cite{GJ07},
there is a finite subset $S$ of $\codomaineff_1$ such that
$\nCSP(\EQp,S)$ --- the ferromagnetic Ising model with local fields, with different spins favoured at different
sites --- is $\BIS$-hard.
Of course similar remarks apply to $\codomainupeff_1$.

In terms of functional clones, the clone $\fclonelim{\EQp,\codomaindown_1}$ is not
amongst those we met in~\S\ref{sec:main}: it is incomparable with $\fclonelim{\NEQ,\codomainbool_1}$,
and strictly contained in $\fclonelim{\IMP,\codomainbool_1}$.  Specifically, we have
\begin{lemma}\label{lem:EQp}
\begin{minipage}[t]{0.75\linewidth}
\begin{enumerate}[topsep=5pt,itemsep=0pt,label=(\roman*)]
\item $\NEQ\notin \fclonelim{\EQp,\codomaindown_1}$.
\item $\EQp\notin\fclonelim{\NEQ,\codomainbool_1}$.
\item $\fclonelim{\EQp,\codomaindown_1}\subset \fclonelim{\EQp,\codomainbool_1}=\fclonelim{\IMP,\codomainbool_1}$, \\[0.25\baselineskip]
and $\fclonelimeff{\EQp,\codomaindowneff_1}\subset \fclonelimeff{\EQp,\codomaineff_1}=\fclonelimeff{\IMP,\codomaineff_1}$.
\end{enumerate}
\end{minipage}
\end{lemma}
\begin{proof}
Recall the class $\calP$ of functions with non-negative Fourier
coefficients defined in \S\ref{sec:cloneLSMk}. Note that
$\calP$ is a \omegadefinable\ functional clone by
Lemma~\ref{lem:lsmk20}.  Note that $\widehat{\EQp}(0,0)=6/4$ and
$\widehat{\EQp}(1,1)=2/4$ and
$\widehat{\EQp}(0,1)=\widehat{\EQp}(1,0)=0$, so $\EQp\in\calP$.  Also,
for any $U\in\codomainbool_1$ we always have $\widehat{U}(0)\geq 0$,
but $\widehat{U}(1)=(U(0)-U(1))/2\geq 0$ if and only if
$U\in\codomaindown_1$.  So
$\fclonelim{\EQp,\codomaindown_1}\subseteq\calP$.

For (i) note that $\widehat{\NEQ}(1,1)=-2/4<0$ so $\NEQ\notin\calP$.

For (ii), Remark~\ref{rem:inclusions} showed that all functions in
$\fclonelim{\NEQ,\codomainbool_1}$ are products of atomic formulas.
Therefore, if $\EQp\in\fclonelim{\NEQ,\codomainbool_1}$, it must have one of
the three forms $U_1(x)U_2(y)$, $U_1(x)\EQ(x,y)$ or $U_1(x)\NEQ(x,y)$, where
$U_1,U_2\in\codomainbool_1$. Now note that $\EQp(x,y)$ is not of any of these.

For (iii), the inclusion $\fclonelim{\EQp,\codomaindown_1}\subseteq \fclonelim{\EQp,\codomainbool_1}$
is trivial.  It is strict since, as we showed above,
$\fclonelim{\EQp,\codomaindown_1}\cap\codomainbool_1=\codomaindown_1$.
The equality follows from Lemma~\ref{lem:binarycases}(iv).
\end{proof}

It is interesting to note that the strict inclusion between
$\fclonelimeff{\EQp,\codomaindown_1}$ and $\fclonelimeff{\EQp,\codomainbool_1}$ is provable,
even though the gap in the computational complexity of the
related counting problems is only suspected.
The other side of the coin is that two functional clones may differ, without there
being a corresponding gap in complexity between the two counting CSPs.
The main result of the section exhibits this phenomenon in a natural
context:  the two functional clones are incomparable, but there is an approximation-preserving
reduction from one of the corresponding counting CSPs to the other.  This is
interesting, as it demonstrates that it is sometimes necessary, when constructing
approximation-preserving reductions, to go beyond the gadgetry implied by the
clone construction (even with the liberal notion employed here, including limits).

Recall that $\oplus_3$ is the relation
$\{(0,0,0),\allowbreak(0,1,1),\allowbreak(1,0,1),\allowbreak(1,1,0)\}$.

\begin{lemma} \label{oplusout_impout} $ \oplus_3 \notin \fclonelimeff{\IMP,\codomaindowneff_1}$
and  $ \IMP \notin \fclonelimeff{\oplus_3,\codomaindowneff_1}$\end{lemma}
\begin{proof}
First we show that $ \oplus_3 \notin \fclonelimeff{\IMP,\codomaindowneff_1}$.
By Lemma~\ref{lem:lsmClosure}, $\fclonelimeff{\IMP,\codomaindowneff_1} \subseteq \LSM$.
However, $\oplus_3 \notin \LSM$, since for $\vecx = (1,1,0)$ and $\vecy = (0,1,1)$,
\[0 = \oplus_3(\vecx \vee \vecy)\> {\oplus_3(\vecx \wedge \vecy)} < \oplus_3(\vecx)\> {\oplus_3(\vecy)}=1.\]

Now we show that $ \IMP \notin \fclonelimeff{\oplus_3,\codomaindowneff_1}$.
Recall the class $\calP$ of functions with non-negative Fourier
coefficients defined in \S\ref{sec:cloneLSMk}. Note that
$\calP$ is a \omegadefinable\ functional clone by
Lemma~\ref{lem:lsmk20}.  For all $U\in\codomaindown_1$ we have
$\widehat{U}(0)=(U(0)+U(1))/2\geq 0$ and
$\widehat{U}(1)=(U(0)-U(1))/2\geq 0$, so $U\in\calP$.  Also,
$\widehat{\oplus_3}=\frac12\EQ_3$ where $\EQ_3$ is the arity 3
equality relation $\{(0,0,0),(1,1,1)\}$. So $\fclonelimeff{\oplus_3,\codomaindowneff_1}\subseteq\calP$. But
$\widehat{\IMP}(0,1)=(1-1-1)/4<0$, so $\IMP\notin\calP$.
\end{proof}

We know now that $\IMP$ is not \omegadef{} in terms of $\oplus_3$ and~$\codomaindowneff_1$. In contrast, we see in the next result that $\IMP$ is nevertheless efficiently reducible to $\oplus_3$ and $\codomaindowneff_1$.

\begin{lemma} \label{lem:obsAPred}
There is a finite subset $S$ of $\codomaindowneff_1$
such that $$\nCSP(\IMP) \APred \nCSP(\oplus_3,S)$$
\end{lemma}
\begin{proof}
First, we need some definitions.
Suppose that~$M$ is a matrix over GF(2) with rows~$V$
and columns~$E$ with $|V|=n$. For a column~$e$ and a ``configuration'' $\sigma:V\rightarrow\B$,
define
${\delta}_e(\sigma)$  to be $  \bigoplus_{i\in V} M_{i,e} \sigma(i)$, where the addition is
over GF(2). $\delta_e(\sigma)$ is the parity of the number of $1$s in column~$e$ of~$M$ that
are assigned to $1$ by~$\sigma$.
Given a parameter $y>0$,
the \emph{Ising partition function} of the binary matroid $\mathcal{M}$ represented by~$M$
is given by
$$
\ZIsing(\calM; y)=
\sum_{\sigma:V\rightarrow \B} \prod_{ e\in E}
{ y}^{1\oplus{\delta}_e(\sigma)}.$$

Now, from \cite[Theorem 3]{trichotomy}, $\nCSP(\IMP)  \APeq \BIS$.
Also, for every efficiently approximable real number $y>1$, from~\cite[Theorem~1]{WE} there is an AP-reduction from $\BIS$ to the problem of computing $\ZIsing(\calM;y)$, for given~$M$.

The set of subsets $A\subseteq E$
such that the submatrix corresponding to~$A$ has
an even number of $1$s in every row is called the \emph{cycle space} of $\calM$ and
is denoted $\calC(\calM)$. A standard calculation expresses $\ZIsing(\calM;y)$ in terms of
$\calC(\calM)$.  Let $w=(y-1)/(y+1)$.

\begin{align*}\hspace{-1truecm}
\allowdisplaybreaks[1]
\sum_{\sigma:V\rightarrow \B} \prod_{ e\in E}
{ y^{1\oplus\delta_e(\sigma)}}
&= \sum_{\sigma} \prod_{e } \left(
\frac{y+1}{2} +  \frac{y-1}{2} {(-1)}^{\delta_e(\sigma)}
\right)
\\
&=
  \left(\frac{y+1}{2}\right)^{|E|}
 \sum_{\sigma} \prod_{e  } \left(
 1 +  w {(-1)}^{\delta_e(\sigma)}
\right)
\\
&=  \left(\frac{y+1}{2}\right)^{|E|}
\sum_{\sigma}
\sum_{A \subseteq E }
\prod_{e\in A}
{ w}  {(-1)}^{ \delta_e(\sigma)}
\\
&=  \left(\frac{y+1}{2}\right)^{|E|}
  \sum_{A \subseteq E }
  w^{|A|}
\sum_{\sigma}
\prod_{ e\in A}   {(-1)}^{ \delta_e(\sigma)}
\\
& = \left(\frac{y+1}{2}\right)^{|E|}
\sum_ {A  \in \calC(\calM)}
w^{|A|}
 2^n,
 \\
& = \left(\frac{y+1}{2}\right)^{|E|} 2^n
\sum_ {A  \in \calC(\calM)}
w^{|A|} .
\end{align*}

Here is the justification of the  penultimate line (why only $A\in\calC(\calM)$ contribute to the sum, and why
the factor $2^n$):
Suppose, for a set~$A \subseteq E$,
that some row~$i$ has
has an odd   number of $1$'s in columns in $A$.
Then for any configuration $\sigma':V\setminus\{i\}\rightarrow \B$,
one of the contributions extending~$\sigma'$
to domain~$V$
contributes~$-1$ and
the other contributes~$+1$.
On the other hand, if $i$ has an even number of $1$'s in~$A$,
then the two contributions are the same, so we just get a factor of~$2$
times the contribution from the smaller problem, without this row.

Note that, since $y>1$, we have $0<w<1$.
Now the point is that it is easy to express
the sum $\sum_ {A  \in \calC(\calM)}
w^{|A|}$ as the solution to an instance of $\nCSP(\oplus_3,U_w)$,
where $U_w$ is the unary function defined by $U_w(0)=1$ and $U_w(1)=w$.
A vector $\vecx$ represents the choice of $A\subseteq E$ --- the $j$'th column is in $A$ iff $x_j=1$.
Then the constraint that the submatrix corresponding to~$A$ has
an even number of $1$s in some row, say row~$i$, is given by the linear equation
$\bigoplus_{j: M_{i,j}=1} x_j=0$. If this linear equation has just two terms then it is an equality, and it can
be represented in the CSP instance by substituting one variable for the other. Otherwise, it can be expressed using conjunctions
of atomic formulas $\oplus_3$.
Thus, we have an AP-reduction from
$\nCSP(\IMP)$ to
  $ \nCSP(\oplus_3,U_ w)$.
  \end{proof}

\begin{remark} \label{lastlabel} Lemma~\ref{oplusout_impout} and~\ref{lem:obsAPred} show that,
as far as counting CSPs are concerned,
the expressibility provided by efficient \omegadef{} functional clones
is more limited than AP-reductions.
Here we show that the problem $\nCSP(\IMP)$ is AP-reducible to a
$\nCSP$ problem whose constraint language consists of functional constraints from $\oplus_3 \cup \codomaindowneff_1$,
but we aren't able to express $\IMP$ using $\oplus_3 \cup \codomaindowneff_1$.
On the other hand, if the definition of efficient \omegadef{} functional clones were somehow extended to
remedy this deficiency, then Lemma~\ref{obsAPclones} would probably have to be weakened.
While we do know (from Lemma~\ref{lem:obsAPred}) that there is a finite subset $S$ of $\codomaindowneff_1$ for which
$\nCSP(\IMP) \APred \nCSP(\oplus_3,S)$,
the corresponding stronger statement from Lemma~\ref{obsAPclones},
$\nCSP(\IMP,\oplus_3,S ) \APred \nCSP(\oplus_3,S)$, is
unlikely to be true since
$\nCSP(\IMP, \oplus_3, S ) \APeq \SAT$
\cite[Theorem 3]{trichotomy}.
\end{remark}

\begin{remark} \label{mylabel}
Lemma~\ref{oplusout_impout}, as with the earlier Lemma~\ref{lem:EQp}, shows that there may be
a rich structure of efficient functional clones
$\fclonelimeff\calF$ with $\codomaindowneff_1\subseteq \calF$.
By contrast, Theorem~\ref{thm:main}
and Lemma~\ref{lem:lsmClosure}
guarantee
that if $\codomaineff_1\subseteq \calF$, then,
the only possibilities are
\begin{itemize}
\item $\fclonelimeff\calF \subseteq  \fclonelimeff{\NEQ,\codomaineff_1}$, or
\item $\fclonelimeff{\IMP,\codomaineff_1} \subseteq \fclonelimeff\calF \subseteq \LSM$,
\item or $\fclonelimeff{\calF} = \codomaineff$.
\end{itemize}
\end{remark}

\section{Using fewer weights}\label{sec:fewWeights}
We saw that a finite set of unary weights suffices to generate the functional
clones that we encountered in previous sections.  Here we observe that just
one or two weights often suffice.

In the following proposition,
$\frac12$ denotes the constant nullary function that takes value~$\frac12$.

\begin{lemma}\label{lem:weights}
\begin{minipage}[t]{0.75\linewidth}
\begin{enumerate}[topsep=5pt,itemsep=0pt,label=(\roman*)]
\item $\codomainnneg_1\subseteq\fclonelim{\IMP,\frac12}$ ($\codomaineff_1\subseteq\fclonelimeff{\IMP,\frac12}$).
\item $\codomainup_1\subseteq\fclonelim{\OR,\frac12}$ ($\codomainupeff_1\subseteq\fclonelimeff{\OR,\frac12}$).
\item $\codomaindown_1\subseteq\fclonelim{\NAND,\frac12}$ ($\codomaindowneff_1\subseteq\fclonelimeff{\NAND,\frac12}$).
\end{enumerate}
\end{minipage}
\end{lemma}

\begin{proof}
Note that (iii) is the same as (ii) with the roles of 0 and 1 reversed, so we will just prove (i) and~(ii).
We start with a general construction that works for both parts (i) and~(ii) of the proposition.
Let $F$ be a binary function and $I,J$ instances of $\nCSP(F)$. We assume the sets of variables of $I$ and $J$ are disjoint. The disjoint sum $I\uplus J$ of $I$ and $J$ is the instance whose set of variables is the union of those of $I$ and $J$, and $F(x,y)$ is in $I\uplus J$ if and only if $F(x,y)$ occurs in $I$ or $J$.
The ordinal sum $I+_\le J$ of $I$ and $J$ is their disjoint sum along with every atomic formula $F(x,y)$ such that $x$ is a variable of $I$ and $y$ is a variable of $J$.

\begin{claim}
\begin{minipage}[t]{0.75\linewidth}
\begin{enumerate}[topsep=5pt,itemsep=0pt,label=(\roman*)]
\item For any $F\in\codomainbool_2$, $Z(I\uplus J)=Z(I)\cdot Z(J)$.
\item If $F\in\{\IMP,\OR\}$ then $Z(I+_\le J)=Z(I)+Z(J)-1$.
\end{enumerate}
\end{minipage}
\end{claim}

The first part is trivial. To show the second part consider an assignment $(\vecx,\vecy)$ such that $\vecx,\vecy$ map the variables of $I$, $J$, respectively, to $\B$, and $F_{I+_\le} J(\vecx,\vecy)\ne0$. If $F=\IMP$ and any of the components of $\vecx$ equals 1,
then $\vecy=\constone$. However, if $\vecx=\constzero$ (or $\vecy=\constone$) then $\vecy$ (resp.\ $\vecx$) can be any legitimate assignment of $J$ (resp.\ $I$). If $F=\OR$ then one of the $\vecx,\vecy$ must be $\constone$, while the remaining one can be any assignment with $F_I(\vecx)\ne0$ or $F_J(\vecy)\ne0$. This completes the proof of the claim.

\smallskip
\def\onevar{\mathbf{2}}
Denote the instance consisting of a single variable without constraints by~$\onevar$, the disjoint sum of $k$~instances $\onevar$ by $\onevar^k$, and the ordinal sum $I+_\le \ldots+_\le I$ of $k$ copies of instance $I$ by $k\cdot I$.
(Note that the operator $+_\le$ is associative, so this makes sense.)
Let also $\mathbf{\frac12}$ denote the instance consisting of a single nullary $\frac12$ function, and let $\mathbf{\frac12}^k$ denote the sum of $k$ copies of $\mathbf{\frac12}$.  Note that $Z(\onevar)=2$ and $Z(\mathbf{\frac12})=\frac12$,
justifying the notation.  Note that for every natural number~$a$ and every positive integer~$\ell$,
$Z(a \cdot {\bf 2}^{\ell}) = a 2^\ell -a + 1$. (This can be proved by induction on~$a$ with base case $a=0$, using Part~(ii)
of the claim for the inductive step.)
Furthermore, for every positive integer~$k$, if $a_1,\ldots,a_k$ are natural numbers
then
$Z(a_1 \cdot {\bf 2}^1 +_{\leq} \cdots +_{\leq} a_k \cdot {\bf 2}^k) = (a_1 2^1 + \cdots + a_k 2^k) - (a_1 + \cdots +a_k) + 1$.
(This can be proved by induction on~$k$ using base case~$k=1$ using the previous observation.)

Suppose $G\in\codomainbool_1$ and let $G^{(n)}$ be a rational valued approximation to~$G$
such that $G^{(n)}(a)\ne0$ and
$|G^{(n)}(a)-G(a)|\leq 2^{-n}$.
Assume that this rational approximation is given as a finite binary expansion, so that
$G^{(n)}(0)=\frac1{2^m}(a_0+a_12^1+\cdots+a_k2^k)\ne0$ and $G^{(n)}(1)=\frac1{2^m}(b_0+b_12^1+\cdots+b_\ell 2^\ell)\ne0$.
Let $I,J$ be instances of $\nCSP(F)$ given by
\begin{align*}
I &= a_1\cdot\onevar^1+_\le\cdots+_\le a_k\cdot\onevar^k +_\le(a_0+a_1+\cdots+a_k-1)\cdot\onevar,\\
J &= b_1\cdot\onevar^1+_\le\cdots+_\le b_\ell\cdot\onevar^\ell+_\le(b_0+b_1+\cdots+b_\ell-1)\cdot\onevar.
\end{align*}
From the observations above,
$Z( (a_0 + \cdots + a_k -1) \cdot {\bf 2}) = a_0 + \cdots + a_k$ so
$Z(I) = 2^m G^{(n)}(0)$. Similarly, $Z(J) = 2^m G^{(n)}(1)$.
Let $V_0,V_1$ be the variables of $I,J$, respectively, and let $C_0$ (respectively, $C_1$) be the set of pairs $(x,y)$ such that $F(x,y)$ is an atomic formula of $I$ (respectively, $J$).

First suppose $F=\IMP$.  Consider the formula
$$
\psi_n=\sum(\mathbf{\tfrac12})^m\left(\prod_{(a,b)\in C_0\cup C_1}\IMP(a,b)\right)\left( \prod_{a\in V_0}\IMP(c,a)\right)\left( \prod_{b\in V_1}\IMP(b,c)\right)
$$
where $c$ is a new variable and the sum is over all variables in
$V_0\cup V_1$.
An assignment $\vecx:V_0\cup V_1\cup\{c\}\to\B$ can only contribute to $F_{\psi_n}$
if either: $\vecx(c)=0$ and $\vecx(b)=0$ for all $b\in V_1$,
or $\vecx(c)=1$ and $\vecx(a)=1$ for all $a\in V_0$.
Hence $F_{\psi_n}(0)=2^{-m}Z(I)=G^{(n)}(0)$ and $F_{\psi_n}(1)=2^{-m}Z(J)=G^{(n)}(1)$.

Now suppose $F=\OR$, and let $G\in\codomainup_1$ and $G^{(n)}$
be as before, but with the restriction $G(1)>G(0)$.
This time,
let $G^{(n)}(0)=\frac1{2^m}(a_0+a_12^1+\cdots+a_k2^k)\ne0$
and $G^{(n)}(1)-G^{(n)}(0)=\frac1{2^m}(b_0+b_12^1+\ldots+b_\ell 2^\ell)\ne0$.
Here we are using the fact that $G^{(n)}(1) - G^{(n)}(0)>0$,
which follows from $G(1)>G(0)$ for sufficiently large $n$.
Let instances $I,J$ be defined for the values $a_0+a_12^1+\cdots+a_k2^k$,
$1+b_0+b_12^1+\cdots+b_\ell 2^\ell$
in the similar way to before (but note the extra 1);
and let $V_0,V_1,C_0,C_1$ again denote the set of variables and constraints of $I,J$. As before $Z(I)=2^mG^{(n)}(0)$
and $Z(J)=2^m(G^{(n)}(1)-G^{(n)}(0))+1$.
Let $K=I+_\le J$ and let $C$ be the set of scopes of the constraints in
$K$.
Consider the formula
$$
\psi_n=\sum(\mathbf{\tfrac12})^m\left(\prod_{(a,b)\in C}\OR(a,b)\right)\left( \prod_{b\in V_1}\OR(b,c)\right),
$$
where $c$ is a new free variable and the sum is over all variables
in $V_0\cup V_1$.  An assignment $\vecx:V_0\cup V_1\cup\{c\}\to\B$ can
only contribute to $F_{\psi_n}$ if either: $\vecx(c)=0$ and
$\vecx(b)=1$ for all $b\in V_0$, or $\vecx(c)=1$.  Therefore
$F_{\psi_n}(0)=2^{-m}Z(I)=G^{(n)}(0)$ and
$F_{\psi_n}(1)=2^{-m}Z(I+_\le J)=2^{-m}(Z(I)+Z(J)-1)=G^{(n)}(1)$.

To obtain the efficient version of the proposition let $M_0$ and $M_1$ be TMs that, given $n$, compute the first $n$ bits of $G(0)$ and $G(1)$ respectively, in time polynomial in $n$. Then a TM $M'$ that constructs $G_\epsilon\in\fclone{F,\frac12}$ such that $\|G_\epsilon-G\|_\infty<\epsilon$ works as follows.
First, it finds the smallest $n$ such that $n>\log\epsilon^{-1}$. Then it runs $M_0$ and $M_1$ on input $n$ to find $G^{(n)}(0)$ and $G^{(n)}(1)$. Finally, $M'$ outputs the formula $\psi_n$. The running time of $M'$ is the sum of: the time to run $M_0$, the time to run $M_1$, and time $O(\log^2\epsilon^{-1})$ to construct $\psi_n$.
\end{proof}

\begin{corollary}\label{non-monotone}
Let $G\in\codomainnneg_1$. (For the results on efficient \ppsdefinability\ assume $G\in\codomaineff_1$.)
\begin{enumerate}[topsep=5pt,itemsep=0pt,label=(\roman*)]
\item If $G(0)>G(1)$ and $G(1)\ne0$ then $\codomainnneg_1\subseteq\fclonelim{\OR,G,\frac12}$,\\[0.25\baselineskip]
and $\codomaineff_1\subseteq\fclonelimeff{\OR,G,\frac12}$.
\item If $G(0)<G(1)$ and $G(0)\ne0$ then $\codomainnneg_1\subseteq\fclonelim{\NAND,G,\frac12}$,\\[0.25\baselineskip]
and $\codomaineff_1\subseteq\fclonelimeff{\NAND,G,\frac12}$.
\end{enumerate}
\end{corollary}

\begin{proof}
We prove (i), as (ii) is quite similar. Let $H$ be a function in~$\codomainbool_1$.
If $H(0)\le H(1)$ then, by Proposition~\ref{lem:weights},
$H\in\fclonelim{\OR,\frac12}$ (or $H\in\fclonelimeff{\OR,\frac12}$). Assume $H(0)>H(1)$. There is $k$ such that $\frac{G(0)^k}{G(1)^k}>\frac{H(0)}{H(1)}$.
Let $H' = H/G^k$.  Then $H' \in \codomainup_1$ so, by Proposition~\ref{lem:weights},
$H' \in \fclonelim{\OR,\frac12}$. Hence $H\in\fclonelim{\OR,G,\frac12}$.
Also, if $G,H\in \codomaineff_1$ then  $H' \in \codomainupeff_1$ so $H' \in \fclonelimeff{\OR,\frac12}$.
Hence $H\in\fclonelimeff{\OR,G,\frac12}$.
\end{proof}

\begin{corollary}\label{cor:or-imp-nand}
$\codomainnneg\subseteq\fclonelim{\OR,\NAND,\frac12}$ and $\codomaineff\subseteq\fclonelimeff{\OR,\NAND,\frac12}$.
\end{corollary}

\begin{proof}
Let $F\in\codomaineff$.
Let $U(x)=\sum_y \NAND(x,y)$. Then $U(0)=2$ and $U(1)=1$ and $U\in\fclonelimeff{\OR,\NAND,\frac12}$.
By Corollary \ref{non-monotone} we have $\codomaineff_1\subseteq\fclonelimeff{\OR,U,\frac12}$, and
by Lemma~\ref{lem:OR},
$F\in\fclonelimeff{\OR,\codomaineff_1}$. So $F\in\fclonelimeff{\OR,\NAND,\frac12}$
by Lemma~\ref{lem:fclonelimeffTransitive}.
\end{proof}

\bibliographystyle{siam}
\bibliography{\jobname}

\begin{thebibliography}{10}

\bibitem{AD}
{\sc R.~Ahlswede and D.~E. Daykin}, {\em An inequality for the weights of two
  families of sets, their unions and intersections}, Z. Wahrsch. Verw. Gebiete,
  43 (1978), pp.~183--185.

\bibitem{Billionnet}
{\sc A.~Billionnet and M.~Minoux}, {\em Maximizing a supermodular pseudoboolean
  function: A polynomial algorithm for supermodular cubic functions}, Discrete
  Applied Mathematics, 12 (1985), pp.~1 -- 11.

\bibitem{playing}
{\sc E.~B{\"o}hler, N.~Creignou, S.~Reith, and H.~Vollmer}, {\em Playing with
  {B}oolean blocks, part {II}: Constraint satisfaction problems}, ACM SIGACT
  Newsletter, 35 (2004), pp.~22--35.

\bibitem{vollmer}
{\sc E.~B{\"o}hler, S.~Reith, H.~Schnoor, and H.~Vollmer}, {\em Bases for
  {B}oolean co-clones}, Inf. Process. Lett., 96 (2005), pp.~59--66.

\bibitem{BorosHammer}
{\sc E.~Boros and P.~L. Hammer}, {\em Pseudo-{B}oolean optimization}, Discrete
  Applied Mathematics, 123 (2002), pp.~155--225.

\bibitem{Bulato11}
{\sc A.~A. Bulatov}, {\em Complexity of conservative constraint satisfaction
  problems}, ACM Trans. Comput. Log., 12 (2011), pp.~24:1--24:66.

\bibitem{BDGJJR12}
{\sc A.~A. Bulatov, M.~Dyer, L.~A. Goldberg, M.~Jalsenius, M.~Jerrum, and
  D.~Richerby}, {\em The complexity of weighted and unweighted \#{CSP}}, J.
  Comput. Syst. Sci., 78 (2012), pp.~681--688.

\bibitem{BuDyGJ12}
{\sc A.~A. Bulatov, M.~Dyer, L.~A. Goldberg, and M.~Jerrum}, {\em
  Log-supermodular functions, functional clones and counting {CSP}s}, in 29th
  International Symposium on Theoretical Aspects of Computer Science, Schloss
  Dagstuhl - Leibniz-Zentrum f\"ur Informatik, 2012, pp.~302--313.

\bibitem{BG}
{\sc A.~A. Bulatov and M.~Grohe}, {\em The complexity of partition functions},
  Theor. Comput. Sci., 348 (2005), pp.~148--186.

\bibitem{CCL11}
{\sc J.-Y. Cai, X.~Chen, and P.~Lu}, {\em Non-negative weighted \#{CSP}s: An
  effective complexity dichotomy}, CoRR, abs/1012.5659 (2010).

\bibitem{CLX11}
{\sc J.-Y. Cai, P.~Lu, and M.~Xia}, {\em Dichotomy for {H}olant* problems of
  {B}oolean domain}, in Proceedings of the 22nd Annual ACM-SIAM Symposium on
  Discrete Algorithms, 2011, pp.~1714--1728.

\bibitem{stablematchings}
{\sc P.~Chebolu, L.~A. Goldberg, and R.~Martin}, {\em The complexity of
  approximately counting stable matchings}, Theoretical Computer Science, 437
  (2012), pp.~35 -- 68.

\bibitem{Chen11}
{\sc X.~Chen}, {\em Guest column: Complexity dichotomies of counting problems},
  SIGACT News, 42 (2011), pp.~54--76.

\bibitem{CohenJeavons}
{\sc D.~Cohen and P.~Jeavons}, {\em Chapter 8: {T}he complexity of constraint
  languages}, in Handbook of Constraint Programming, vol.~2 of Foundations of
  Artificial Intelligence, Elsevier, 2006, pp.~245--280.

\bibitem{CKS}
{\sc N.~Creignou, S.~Khanna, and M.~Sudan}, {\em Complexity Classifications of
  {B}oolean Constraint Satisfaction Problems}, SIAM, Philadelphia, PA, USA,
  2001.

\bibitem{CKZ}
{\sc N.~Creignou, P.~Kolaitis, and B.~Zanuttini}, {\em Structure identification
  of {B}oolean relations and plain bases for co-clones}, Journal of Computer
  and System Sciences, 74 (2008), pp.~1103--1115.

\bibitem{deWolf}
{\sc R.~de~Wolf}, {\em A Brief Introduction to Fourier Analysis on the Boolean
  Cube}, no.~1 in Graduate Surveys, Theory of Computing Library, 2008.

\bibitem{DGGJ}
{\sc M.~Dyer, L.~A. Goldberg, C.~Greenhill, and M.~Jerrum}, {\em The relative
  complexity of approximate counting problems}, Algorithmica, 38 (2004),
  pp.~471--500.

\bibitem{wbool}
{\sc M.~Dyer, L.~A. Goldberg, and M.~Jerrum}, {\em The complexity of weighted
  {B}oolean \#{CSP}}, SIAM Journal on Computing, 38 (2009), pp.~1970--1986.

\bibitem{trichotomy}
\leavevmode\vrule height 2pt depth -1.6pt width 23pt, {\em An approximation
  trichotomy for {B}oolean \#{CSP}}, Journal of Computer and System Sciences,
  76 (2010), pp.~267--277.

\bibitem{GJ07}
{\sc L.~A. Goldberg and M.~Jerrum}, {\em The complexity of ferromagnetic
  {I}sing with local fields}, Combin. Probab. Comput., 16 (2007), pp.~43--61.

\bibitem{GoldbergJ10}
{\sc L.~A. Goldberg and M.~Jerrum}, {\em Approximating the partition function
  of the ferromagnetic {P}otts model}, in Automata, Languages and Programming,
  37th International Colloquium, Proceedings, Part I, vol.~6198 of Lecture
  Notes in Computer Science, Springer, 2010, pp.~396--407.

\bibitem{WE}
\leavevmode\vrule height 2pt depth -1.6pt width 23pt, {\em Approximating the
  {T}utte polynomial of a binary matroid and other related combinatorial
  polynomials}, CoRR, abs/1006.5234 (2010).

\bibitem{GrMaRo00}
{\sc M.~Grabisch, J.-L. Marichal, and M.~Roubens}, {\em Equivalent
  representations of set functions}, Math. Oper. Res., 25 (2000), pp.~157--178.

\bibitem{JS93}
{\sc M.~Jerrum and A.~Sinclair}, {\em Polynomial-time approximation algorithms
  for the {I}sing model}, SIAM J. Comput., 22 (1993), pp.~1087--1116.

\bibitem{JVV86}
{\sc M.~Jerrum, L.~G. Valiant, and V.~Vazirani}, {\em Random generation of
  combinatorial structures from a uniform distribution}, Theoret. Comput. Sci.,
  43 (1986), pp.~169--188.

\bibitem{KoFri82}
{\sc K.-I. Ko and H.~Friedman}, {\em Computational complexity of real
  functions}, Theoretical Computer Science, 20 (1982), pp.~323--352.

\bibitem{KolZiv12}
{\sc V.~Kolmogorov and S.~\v{Z}ivn\'y}, {\em The complexity of conservative
  valued {CSP}s}, in Proceedings of the 23rd Annual ACM-SIAM Symposium on
  Discrete Algorithms, SIAM, 2012, pp.~750--759.

\bibitem{McQuil11}
{\sc C.~McQuillan}, {\em {LSM} is not generated by binary functions}, CoRR,
  abs/1110.0461 (2011).

\bibitem{MU05}
{\sc M.~Mitzenmacher and E.~Upfal}, {\em Probability and Computing}, Cambridge
  University Press, Cambridge, 2005.

\bibitem{Topkis}
{\sc D.~M. Topkis}, {\em Minimizing a submodular function on a lattice},
  Operations Research, 26 (1978), pp.~305--321.

\bibitem{Tora87}
{\sc J.~Tor\'an}, {\em On the complexity of computable real sequences}, RAIRO
  Inform. Th\'eor. Appl., 21 (1987), pp.~175--180.

\bibitem{ZJ10}
{\sc S.~\v{Z}ivn\'y and P.~Jeavons}, {\em Classes of submodular constraints
  expressible by graph cuts}, Constraints, 15 (2010), pp.~430--452.

\bibitem{Weg}
{\sc I.~Wegener}, {\em Complexity Theory}, Springer-Verlag, Berlin, 2005.

\bibitem{Tomo}
{\sc T.~Yamakami}, {\em Approximate counting for complex-weighted {B}oolean
  constraint satisfaction problems}, CoRR, abs/1007.0391 (2010).

\bibitem{ZCJ09}
{\sc S.~{\v{Z}}ivn{\'y}, D.~Cohen, and P.~Jeavons}, {\em The expressive power
  of binary submodular functions}, Discrete Appl. Math., 157 (2009),
  pp.~3347--3358.

\end{thebibliography}

\end{document}